\documentclass[journal]{IEEEtran}

\usepackage{times,amsmath,color,amssymb,graphicx,epsfig,cite,psfrag,enumerate}
\usepackage{subfigure,multirow,bm,stfloats}
\newtheorem{Remark}{\it Remark}[section]

\newtheorem{Proposition}{\it Proposition}[section]
\newtheorem{Lemma}{\it Lemma}[section]

\newcounter{Mytempeqncnt1}

\begin{document}
%\linespread{1}
% paper title
% can use linebreaks \\ within to get better formatting as desired
\title{Throughput Maximization for the Gaussian Relay Channel with Energy Harvesting Constraints}

%\setlength{\baselineskip}{17.4pt}
% author names and affiliations
% use a multiple column layout for up to three different
% affiliations

%\author{\IEEEauthorblockN{Chuan Huang,~\IEEEmembership{Student Member,~IEEE,}~Jinhua Jiang,~\IEEEmembership{Member,~IEEE,} \\ Shuguang
%Cui},~\IEEEmembership{Member,~IEEE,}
%
%\thanks{Chuan Huang and Shuguang Cui are with the Department of Electrical
%and Computer Engineering, Texas A\&M University, College Station,
%TX, 77843. Emails: \{huangch, cui\}@tamu.edu.}
%
%\thanks{Jinhua Jiang is with the Department of Electrical Engineering,
%Stanford University, Stanford, CA, 94305, Email:
%jhjiang@stanford.edu.}}

\author{\IEEEauthorblockN{Chuan Huang,~\IEEEmembership{Student Member,~IEEE},~Rui Zhang,~\IEEEmembership{Member,~IEEE},~Shuguang
Cui},~\IEEEmembership{Member,~IEEE}

\thanks{Manuscript received in August 29, 2011; revised in February 28, 2012. Parts of this paper have been presented in IEEE ICASSP'2012 and IEEE ICC'2012. This work is supported in part by DoD-DTRA under Grant HDTRA1-08-1-0010.}

\thanks{Chuan Huang and Shuguang Cui are with the Department of Electrical
and Computer Engineering, Texas A\&M University, College Station,
TX, 77843. Emails: \{huangch, cui\}@tamu.edu.

Rui Zhang is with the Department of Electrical and Computer Engineering, National University of Singapore, Singapore 117576. Email: elezhang@nus.edu.sg.
}}

%\thanks{Digital Object Identifier}}

%\markboth{IEEE JOURNAL ON SELECTED AREAS IN COMMUNICATIONS,~VOL.~6, NO.~1, January~2012}%
%{HUANG \MakeLowercase{\textit{et al.}}: THROUGHPUT MAXIMIZATION FOR THE GAUSSIAN RELAY CHANNEL WITH ENERGY HARVESTING CONSTRAINTS}

%\IEEEpubid{0000-0000/00\$00.00~\copyright~2012 IEEE}

% make the title area
\maketitle
\begin{abstract}
%\boldmath
This paper considers the use of energy harvesters, instead of conventional time-invariant energy sources, in wireless cooperative communication. For the purpose of exposition, we study the classic three-node Gaussian relay channel with decode-and-forward (DF) relaying, in which the source and relay nodes transmit with power drawn from energy-harvesting (EH) sources. Assuming a deterministic EH model under which the energy arrival time and the harvested amount are known prior to transmission, the throughput maximization problem over a finite horizon of $N$ transmission blocks is investigated. In particular, two types of data traffic with different delay constraints are considered: delay-constrained (DC) traffic (for which only one-block decoding delay is allowed at the destination) and no-delay-constrained (NDC) traffic (for which arbitrary decoding delay up to $N$ blocks is allowed). For the DC case, we show that the joint source and relay power allocation over time is necessary to achieve the maximum throughput, and propose an efficient algorithm to compute the optimal power profiles. For the NDC case, although the throughput maximization problem is non-convex, we prove the optimality of a separation principle for the source and relay power allocation problems, based upon which a two-stage power allocation algorithm is developed to obtain the optimal source and relay power profiles separately. Furthermore, we
compare the DC and NDC cases, and obtain the sufficient and
necessary conditions under which the NDC case performs strictly
better than the DC case. It is shown that NDC transmission is able
to exploit a new form of diversity arising from the independent
source and relay energy availability over time in cooperative
communication, termed ``energy diversity'', even with time-invariant
channels.
\end{abstract}


% IEEEtran.cls defaults to using nonbold math in the Abstract.
% This preserves the distinction between vectors and scalars. However,
% if the conference you are submitting to favors bold math in the abstract,
% then you can use LaTeX's standard command \boldmath at the very start
% of the abstract to achieve this. Many IEEE journals/conferences frown on
% math in the abstract anyway.

% no keywords
\section{Introduction}

\IEEEPARstart{I}{n} conventional energy-constrained wireless communication systems such as wireless sensor networks (WSNs), sensors are equipped with fixed energy supply devices, e.g., batteries, which have limited operation time. When thousands of sensors are deployed in a hostile or toxic environment, recharging or replacing batteries becomes inconvenient and even impossible. Hence, harvesting energy from the environment is a much easier and safer way to provide almost unlimited energy supply for WSNs. However, compared with conventional time-invariant energy sources, energy replenished by harvesters is intermittent over time, e.g., energy fluctuation caused by time-dependent solar and wind patterns. As a result, wireless devices powered by renewable energy are subject to the energy-harvesting (EH) constraints over time, i.e., the total energy consumed up to any time must be less than the energy harvested by that time.
\IEEEpubidadjcol

Wireless communication with EH nodes has recently drawn significant research attention. In \cite{niyato,kansal}, the authors investigated the power management strategies for WSNs with EH nodes, for which random EH models were assumed. For the point-to-point communication powered by EH sources, the power management problem was studied in \cite{yang,yener} with the deterministic EH model, and in \cite{sharma,rui,ozel} with the random EH model. In particular, with the deterministic EH model, under which the energy amount and arrival time are assumed to be known prior to transmission, the authors in \cite{yang} studied the throughput maximization and transmission time minimization problems for the point-to-point additive white Gaussian noise (AWGN) channel. These results were generalized in \cite{yener} by further considering the finite energy storage limit. With the block Markov random EH model, the authors in \cite{rui} studied the throughput maximization problem over the fading AWGN channel and derived the optimal power allocation polices via dynamic programming and convex optimization techniques. In addition, with an independent and identically distributed (i.i.d.) EH model, the authors in \cite{ozel} studied the AWGN channel capacity under the EH constraints, and showed that even with the time-varying energy source, the same capacity can be achieved as that for the conventional case of constant power supply with the same total transmission energy consumed. It is worth noting that the authors in \cite{gunduz} considered the throughput maximization problem for the Gaussian two-hop relay channel without considering the direct link between the source and the destination, which is a special case for the relay channel model considered in our paper.

On the other hand, node cooperation has been known as an effective way to improve the
system capacity and diversity performance in wireless networks. With conventional time-invariant energy sources, the full-duplex relay channel has been thoroughly investigated in, e.g., \cite{cover,xie,kramer,cover2}, where various
achievable rates with decode-and-forward (DF) and compress-and-forward (CF) relaying schemes were obtained. For the half-duplex relay channel in which the relay needs to transmit and receive over orthogonal time slots or frequency bands, the achievable rates and power allocation polices have been examined in \cite{madsen}. In particular, the orthogonal half-duplex relay channel, in which the relay-destination link is orthogonal to the source-relay and source-destination links, has been studied in \cite{liang}.

In this paper, we study the half-duplex orthogonal Gaussian relay channel with EH source and relay nodes, as shown in Fig. \ref{system_model}. It is assumed that the relay transmits and receives over two different frequency bands, and thus the relay-destination link is orthogonal to both the source-relay and source-destination links. Here we consider the simple case with deterministic source and relay energy profiles, corresponding to practical scenarios where the EH level can be predicted with negligible errors, and leave the more general random cases for future study. Moreover, we focus on the DF relaying scheme for the purpose of exposition. We examine the throughput maximization problems over a finite horizon of $N$-block transmission\footnote{Note that in total $(N+1)$-block time is needed for each $N$-block transmission due to the one-block decoding delay at the relay.}. In each block, the source transmits a new message, which is received and decoded by the relay, and then forwarded to the destination in the subsequent one or more blocks. Our main objective is to study the structure of the optimal power and rate allocation at the source and the relay over different blocks to maximize the total throughput, under individual source and relay EH constraints. Specifically, we consider the following two types of data traffic with different decoding delay requirements at the destination:
\begin{enumerate}
  \item Delay-constrained (DC) traffic: The destination is required to decode the $i$-th source message, $i=1,\cdots,N$, immediately after it receives the signals from the source in the $i$-th block and from the relay in the $(i+1)$-th block. With such a requirement, the relay needs to forward the source message received in one block to the destination immediately in the next block;
  \item No-delay-constrained (NDC) traffic: The destination can tolerate arbitrary decoding delays provided that all source messages are decoded at the end of each $N$-block transmission. Consequently, the relay is allowed to store the decoded source message of the $i$-th block, and forward it to the destination in any of the remaining $(i+1)$-th, $\cdots$, $(N+1)$-th blocks.
\end{enumerate}
Clearly, the NDC case allows more flexible relay operations than the
DC case, and is thus expected to achieve a larger throughput in
general. It is worth noting that in practical EH systems, the source
and relay may have independent energy arrivals over time; as a
result, there exists a new form of diversity, termed ``energy
diversity'', to be exploited in cooperative communication with EH
nodes, even for the case of time-invariant channels. As will be
shown in this paper, the NDC transmission is able to exploit the
energy diversity to achieve a larger throughput than the DC
counterpart, thanks to the more relaxed delay requirement. The main
contributions of this paper are summarized as follows:
\begin{enumerate}
  \item For the DC case, we formulate a convex throughput maximization problem, and develop a joint source and relay optimal power allocation algorithm by exploiting the Karush-Kuhn-Tucker (KKT) optimality conditions and the monotonic property of the optimal power allocation that is non-decreasing over time. It is shown that the developed algorithm is a forward search in time over the two-dimensional (at both source and relay) harvested energy profiles, which can be considered as an extension of the one-dimensional (at source only) search algorithm in \cite{yang,rui} for the case of point-to-point AWGN channel.
  \item For the NDC case, although the throughput maximization problem is in general non-convex, a separation principle for the source and relay power allocation problem is proved to be optimal, upon which the original problem is decoupled into two convex subproblems that separately achieve the optimal source and relay power allocation. Such optimal source and relay power allocation is shown to be also non-decreasing over time, similar to the DC case. Moreover, we derive the necessary and sufficient conditions for the NDC case to strictly improve the system throughput over the DC case.
\end{enumerate}

The remainder of the paper is organized as follows. Section II presents the system model and summarizes the main assumptions in this paper. Section III formulates the throughput maximization problems for the DC and NDC cases, respectively. Sections IV and V develop algorithms to solve the formulated problems for the DC and NDC cases, respectively. Numerical results are presented in Section VI to validate the theoretical results. Finally, Section VII concludes the paper.

{\it Notation}: $\log(\cdot)$ and $\ln(\cdot)$ stand for the base-2 and natural logarithms, respectively; $\mathcal{C}(x) = \frac{1}{2}\log \left( 1+ x\right)$ denotes for the AWGN channel capacity; $\min \left\{x,y \right\}$ and $\max \left\{x,y \right\}$ denote the minimum and maximum between two real numbers $x$ and $y$, respectively; $\left(x\right)^+ = \max(0,x)$.

\section{System Model}

We consider the classic three-node relay channel, which consists of one source-destination pair and one relay, as shown in Fig. \ref{system_model}. We assume that the relay node operates in a half-duplex mode over two orthogonal frequency bands, while the source-relay and source-destination use the same band. For simplicity, we do not consider the bandwidth allocation problem for the relay, and assume that the source-relay and relay-destination links operate with equal bandwidth.

\begin{figure}[!t]
\centering
\includegraphics[width=.9 \linewidth]{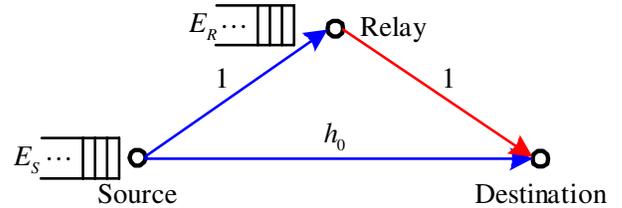}
\caption{Orthogonal relay channel with energy harvesting source and relay nodes.}
\label{system_model}
\end{figure}

\begin{figure*}[!b]
\vspace{4pt}
\hrulefill
\normalsize
\setcounter{Mytempeqncnt1}{\value{equation}}
\setcounter{equation}{12}
\begin{align}
(\text{P1}) & ~\max_{\left\{P_S(i) \right\},\left\{P_R(i+1)\right\}} ~~\frac{1}{2(N+1)} \sum_{i=1}^{N}  \min \left\{ \mathcal{C} \left( P_S(i) \right), \mathcal{C} \left(h_0 P_S(i) \right) + \mathcal{C} \left( P_R(i+1) \right) \right\}  \label{no_buffer_opt} \\
\text{s. t.} &~(\ref{source_relay_energy-const1}),~(\ref{source_relay_energy-const2}),~P_S(i) \geq 0,~P_R(i+1) \geq 0,i=1,\cdots,N. \label{no_buffer_opt2}
\end{align}
\setcounter{equation}{\value{Mytempeqncnt1}}
\end{figure*}

We consider the DF relaying scheme, which requires the relay to successfully decode the source message. Moreover, we adopt an $N$-block transmission protocol described as follows: During each of the $N$ source transmission blocks, say, the $i$-th block, $1 \leq i \leq N$, the source transmits a new message $w_i$ with power $P_S(i)$ and rate $R(i)$; upon receiving the signal from the source, the relay decodes $w_i$, and generates a binning index for $w_i$ based on the ``random binning'' technique \cite{liang} with rate $R_B(i+1)$. In the $(i+1)$-th block, the relay transmits a message $v_{i+1}$ with power $P_R(i+1)$ and rate $\mathcal{C} \left( P_R(i+1) \right)$. It is noted that for the DC case (defined in Section I), $v_{i+1}$ is the binning index of source message $w_i$ only; while for the NDC case (defined in Section I), $v_{i+1}$ may contain the information of binning indices for all source messages $w_k$'s, $k \leq i$. Moreover, we assume that each block has $B$ channel uses, where $B$ is assumed large enough such that the channel capacity results in \cite{liang,cover} are good approximations to the communication rates in practical systems.

In addition to the block transmission model, we assume that the harvested energy arrives at the beginning of each block with known amounts $E_S(i)$ in the $i$-th block and $E_R(i+1)$ in the $(i+1)$-th block, $i=1,2,\cdots,N$, at the source and the relay, respectively. In this paper, we assume that the battery capacity to store the harvested energy is infinite, and the consumed energy at the source or relay other than transmission energy is small and thus negligible. Thus, the amount of energy available for each block transmission is constrained by the following source and relay EH constraints:
\begin{align}
\sum_{i=1}^k P_S(i) &  \leq \frac{1}{ B} \sum_{i=1}^k E_S(i),~k=1,\cdots,N, \label{source_relay_energy-const1} \\
\sum_{i=1}^k P_R(i+1) & \leq \frac{1}{ B} \sum_{i=1}^k E_R(i+1),~k=1,\cdots,N. \label{source_relay_energy-const2}
\end{align}

For the $i$-th source and the $(i+1)$-th relay transmission blocks\footnote{Note that the $(i+1)$-th relay transmission block in fact corresponds to the $i$-th source message in the DC case.}, $i=1,\cdots,N$, the channel input-output relationships are given as:
\begin{align}
y_{sr} (i) & = \sqrt{h_{sr}} x_s (i) + n_r (i), \label{channel1}\\
y_{sd} (i) & = \sqrt{h_{sd}} x_s (i) + n_d (i) , \\
y_{rd} (i+1) & = \sqrt{h_{rd}} x_r (i+1) + w_d (i+1), \label{channel3}
\end{align}
where $x_s (i)$ and $x_r (i+1)$ are the transmitted signals in the $i$-th source and the $(i+1)$-th relay transmission blocks with power $P_S(i)$ and $P_R(i+1)$, respectively; $y_{sr} (i)$ is the received signal at the relay; $y_{sd} (i)$ and $y_{rd} (i+1)$ are the received signals at the destination from the source and the relay, respectively; $h_{sr}$, $h_{rd}$, and $h_{sd}$ are the constant channel power gains for the source-relay, relay-destination, and source-destination links, respectively;
$n_r (i)$, $n_d (i)$, and $w_d (i+1)$ are i.i.d. circularly symmetric complex Gaussian (CSCG) noises each with zero mean and unit variance.

With the above assumptions, the received signal-to-noise ratio (SNR) for the source-relay, source-destination, and relay-destination links are given as $\gamma_{sr}(i) = P_S(i) h_{sr}$, $\gamma_{sd}(i) = P_S(i) h_{sd}$, and $\gamma_{rd}(i+1) = P_R(i+1) h_{rd}$, respectively. Define new source/relay energy and power profiles as $\widetilde{E}_S(i) = E_S(i)h_{sr}$, $\widetilde{E}_R(i+1) = E_R(i+1)h_{rd}$, $\widetilde{P}_S(i) = P_S(i)h_{sr}$, and $\widetilde{P}_R(i+1) = P_R(i+1)h_{rd}$, and new channel gains as $\widetilde{h}_{sr} = \widetilde{h}_{rd}=1$ and $\widetilde{h}_{sd} = \frac{h_{sd}}{h_{sr}} \doteq h_0$. It is easy to check that with the new defined parameters, each link has the same SNR as before and the EH constraints given in (\ref{source_relay_energy-const1})-(\ref{source_relay_energy-const2}) are satisfied accordingly with the new power and energy profiles. As such, we could always determine the source and relay power profiles $\widetilde{P}_S(i)$'s and $\widetilde{P}_R(i+1)$'s first under the assumptions that $\widetilde{h}_{sr} = \widetilde{h}_{rd}=1$ and $\widetilde{h}_{sd} = \frac{h_{sd}}{h_{sr}} $, and then obtain $P_S(i)$'s and $P_R(i+1)$'s by scaling accordingly. Therefore, for notation simplicity and without loss of generality, we simplify the signal models in (\ref{channel1})-(\ref{channel3}) as
\begin{align}
y_{sr} (i) & =  x_s (i) + n_r (i), \\
y_{sd} (i) & = \sqrt{h_0} x_s (i) + n_d (i) , \\
y_{rd} (i+1) & =  x_r (i+1) + w_d (i+1),
\end{align}
by setting $h_{sr} = h_{rd}= 1$ and $h_{sd}=h_0$.

Moreover, it is assumed that $0 \leq h_0<1$, which means that the source-relay link is stronger than the source-destination link. Thus, the relay can always help with increasing the achievable rate from the source to the destination.

\begin{figure*}[!b]
\vspace{4pt}
\hrulefill
\normalsize
\setcounter{Mytempeqncnt1}{\value{equation}}
\setcounter{equation}{23}
\begin{align}
(\text{P2}) & \max_{\left\{P_S(i)\right\}, \left\{P_R(i+1)\right\} } \frac{1}{2(N+1)} \sum_{i=1}^N \mathcal{C} \left(h_0 P_S (i) \right) + \mathcal{C} \left( P_R (i+1) \right)  \label{inf_buffer_DF1} \\
\text{s. t.} & \sum_{i=1}^k \mathcal{C} \left( h_0 P_S (i) \right) +  \mathcal{C} \left( P_R (i+1) \right) \leq \sum_{i=1}^k \mathcal{C} \left( P_S (i) \right) ,k=1,\cdots,N, \text{and}~(\ref{no_buffer_opt2}). \label{inf_buffer_DF2}
\end{align}
\setcounter{equation}{\value{Mytempeqncnt1}}
\end{figure*}

\section{Problem Formulation}

\subsection{Delay-Constrained Case}

First consider the DC case. Since in the $i$-th source transmission block, the source transmits message $w_i$ with power $P_S(i)$ and rate $R(i)$, the relay decodes $w_i$ reliably only if
\begin{align} \label{DC_rate1}
R(i) \leq  \mathcal{C}(P_S(i)).
\end{align}
Then in the next block, the relay partitions $w_i$ into bins with an equivalent rate $ R_B(i+1)$ \cite{liang}, and transmits the binning index in message $v_{i+1}$ to the destination with power $P_R(i+1)$. At the destination, it first decodes $v_{i+1}$, if
\begin{align} \label{orth_DF_rate2}
R_B(i+1) \leq \mathcal{C} (P_R(i+1)),
\end{align}
and then decodes the original message $w_i$ if
\begin{align}
R(i) & \leq  \mathcal{C} (h_0 P_S(i)) +  R_B(i+1) \nonumber \\
 & \leq  \mathcal{C} (h_0 P_S(i)) + \mathcal{C} (P_R(i+1)), \label{DC_rate2}
\end{align}
where the second inequality is due to (\ref{orth_DF_rate2}). From (\ref{DC_rate1}) and (\ref{DC_rate2}), the achievable rate for the $i$-th source message is given by
\begin{align} \label{orth_DF_rate}
R(i)= \min \left\{ \mathcal{C} \left(P_S(i)  \right), \mathcal{C} \left( h_0 P_S(i) \right) + \mathcal{C} \left( P_R(i+1)  \right) \right\}.
\end{align}
Note that for the case of $h_0 = 0$, the coding scheme for the relay can be simplified to repetition coding, i.e., the source and the relay can use the same codebook.

Considering the $N$-block transmission, the average throughput in the unit of bits/sec/Hz (bps/Hz) is maximized by solving Problem (P1) in (\ref{no_buffer_opt})-(\ref{no_buffer_opt2}), where the factor $\frac{1}{2}$ in (\ref{no_buffer_opt}) is due to half-duplex relaying, and $\frac{1}{N+1}$ is due to the fact that each $N$-block transmission requires $(N+1)$-block duration. Next, some properties of the optimal power allocation solution for Problem (P1) are revealed.
\setcounter{equation}{14}

\begin{Proposition} \label{no_buffer_cons}
For $0 \leq h_0 <1$, there exist optimal power profiles $P_S^*(i)$'s and $P_R^*(i+1)$'s, which achieve the maximum throughput of Problem (P1) and satisfy the following inequalities:
\begin{align} \label{no_buffer_cons_equa}
\mathcal{C} \left(P_S^*(i) \right) \geq \mathcal{C} \left(h_0 P_S^*(i) \right) + \mathcal{C} \left( P_R^*(i+1) \right),~i=1,\cdots,N;
\end{align}
moreover, for the case of $h_0 = 0$, there exist optimal power profiles for Problem (P1) with
\begin{align} \label{no_buffer_power_nodire}
P_S^*(i) =  P_R^*(i+1),~i=1,\cdots,N.
\end{align}
\end{Proposition}
\begin{proof}
If (\ref{no_buffer_cons_equa}) is not satisfied for any $i$, we can always decrease $P_R^*(i+1)$ until it is satisfied, without reducing the achievable rate of the $i$-th source message. Similarly, if (\ref{no_buffer_power_nodire}) is not satisfied at any $i$, we can reduce $P_S^*(i)$ if $P_S^*(i)> P_R^*(i+1)$ or $P_R^*(i+1)$ if $P_R^*(i+1)> P_S^*(i)$ until the equality holds, without reducing the rate of the $i$-th source message. Thus, this proposition is proved.
\end{proof}

\begin{Remark} \label{DC_opt_not_unique}
From Proposition \ref{no_buffer_cons}, we infer that the optimal solution of Problem (P1) is not unique in general, e.g., when the energy harvested at the relay is excessively larger than that at the source. In the sequel, we are only interested in finding the optimal solutions for Problem (P1) satisfying (\ref{no_buffer_cons_equa}) and (\ref{no_buffer_power_nodire}) for the cases of $0<h_0<1$ and $h_0=0$, respectively, which achieve the minimum energy consumptions at the source and relay.
\end{Remark}

By (\ref{no_buffer_cons_equa}) and (\ref{orth_DF_rate2}), we obtain that
\begin{align*}
\mathcal{C} \left(P_S^*(i) \right) & \geq \mathcal{C} \left(h_0 P_S^*(i) \right) + \mathcal{C} \left(P_R^*(i+1) \right) \\
& \geq \mathcal{C} \left(h_0 P_S^*(i) \right) + R_B^*(i+1),~i=1,\cdots,N.
\end{align*}
Together with (\ref{orth_DF_rate2}), it follows that
\begin{align}
R_B(i+1)& = \min \left\{ \mathcal{C}(P_R(i+1)), \mathcal{C}(P_S(i)) -  \mathcal{C}(h_0 P_S(i)) \right\}, \nonumber \\
&~~~~~~~~~~~~~~~~~~~~~~~~~~~~~~~~i=1,\cdots,N. \label{DC_binning}
\end{align}
As such, if we can solve Problem (P1), by further applying (\ref{DC_binning}), we can obtain the optimal relay rate allocation for the DC case.

It is easy to verify that Problem (P1) is convex, and thus solvable by existing convex optimization techniques, e.g., the interior point method \cite{boyd}. However, such an approach does not reveal any insight for the optimal solution. Thus, in this paper, we develop an alternative method to solve Problem (P1) by exploiting its special structure, as will be shown later in Section IV.

\begin{figure*}[!b]
\vspace{4pt}
\hrulefill
\normalsize
\setcounter{Mytempeqncnt1}{\value{equation}}
\setcounter{equation}{25}
\begin{align}
(\text{P$1^*$})&~\max_{\left\{ P_S(i)\right\},\left\{P_R(i+1)\right\}} ~~\frac{1}{2(N+1)} \sum_{i=1}^{N}   \mathcal{C} \left(h_0 P_S(i) \right) + \mathcal{C} \left( P_R(i+1) \right) \label{new2111} \\
\text{s. t.}  &~~\mathcal{C} (h_0 P_S(i)) + \mathcal{C} ( P_R(i+1))  \leq \mathcal{C} ( P_S(i)),i=1,\cdots,N,\text{and}~(\ref{no_buffer_opt2}). \label{new2}
\end{align}
\begin{align}
&\mathcal{L}\left( P_S(i),P_R(i+1),\mu_k,\lambda_{k},\gamma_i,\eta_{i+1} \right) \nonumber = \frac{1}{2(N+1)} \sum_{i=1}^{N} \min \left\{ \mathcal{C} \left( P_S(i) \right), \mathcal{C} \left( h_0 P_S(i) \right) + \mathcal{C} \left( P_R(i+1) \right) \right\} \nonumber \\
& ~~~ - \sum_{k=1}^{N} \mu_k \left(\sum_{i=1}^{k} B  P_S(i) -  E_S(i)  \right)
- \sum_{k=1}^{N} \lambda_k \left( \sum_{i=1}^{k}  B  P_R(i+1) -  E_R(i+1)  \right) + \sum_{i=1}^{N} \gamma_i P_S(i) + \sum_{i=1}^{N} \eta_{i+1} P_R(i+1). \label{no_buffer_lag}
\end{align}
\begin{align}
\frac{\partial \mathcal{L}}{ \partial P_S(i)} & = \left\{
                                                    \begin{array}{ll}
                                                      \frac{1}{4(N+1)} \times \frac{1}{1+P_S(i)} , & P_R(i+1) \geq \frac{(1-h_0)P_S(i)}{1+ h_0 P_S(i)} \\
                                                      \frac{1}{4(N+1)}\times \frac{h_0}{1+h_0P_S(i)} , & \text{otherwise}
                                                    \end{array}
                                                  \right. - B \sum_{k=i}^{N} \mu_k + \gamma_i,  \label{deter_with_lag1}\\
\frac{\partial \mathcal{L}}{ \partial P_R(i+1)} & = \left\{
                                                    \begin{array}{ll}
                                                      0, & P_R(i+1) \geq \frac{(1-h_0)P_S(i)}{1+ h_0 P_S(i)} \\
                                                      \frac{1}{4(N+1)} \times \frac{1}{1+P_R(i+1)} , & \text{otherwise}
                                                    \end{array}
                                                  \right. - B \sum_{k=i}^{N} \lambda_k + \eta_{i+1}. \label{deter_with_lag2}
\end{align}
\setcounter{equation}{\value{Mytempeqncnt1}}
\end{figure*}

\subsection{No-Delay-Constrained Case}

For the NDC case, the relay operates the same as the DC case except that it is allowed to transmit the binning index for message $w_i$ in messages $v_{i+1},\cdots,v_{N+1}$ instead of $v_{i+1}$ only as in the DC case. At the destination, the binning indices for all source messages can be successfully decoded if
\begin{align}
 \sum_{i=1}^N R_B(i+1) & = \sum_{i=1}^N  \mathcal{C} ( P_R(i+1)) , \\
 \sum_{i=k}^N R_B(i+1) &  \leq \sum_{i=k}^N \mathcal{C} ( P_R(i+1)) ,~2 \leq k \leq N,
\end{align}
which is equivalent to
\begin{align}
\sum_{i=1}^k R_B(i+1) &\geq \sum_{i=1}^k  \mathcal{C} ( P_R(i+1)),~1 \leq k \leq N-1, \nonumber \\
\sum_{i=1}^N R_B(i+1) & = \sum_{i=1}^N  \mathcal{C} ( P_R(i+1)). \label{NDC_const2}
\end{align}
With the decoded binning index, the $i$-th source message can be decoded successfully at the destination if $ R(i ) \leq \mathcal{C} (h_0 P_S(i)) + R_B(i+1),~i=1,\cdots,N$. Combining this with (\ref{DC_rate1}), the achievable rate of the $i$-th source message is given as
\begin{align}
R(i) & = \min \left\{ \mathcal{C} (P_S(i)), \mathcal{C} (h_0 P_S(i)) + R_B(i+1)  \right\} \nonumber \\
&= \mathcal{C} (h_0 P_S(i)) + R_B(i+1),~i=1,\cdots,N, \label{NDC_rate1}
\end{align}
where the second equality is due to a similar argument as Proposition \ref{no_buffer_cons}. Note that for the special case of $h_0=0$, we have $R(i)=R_B(i+1)$ in (\ref{NDC_rate1}). In addition, (\ref{NDC_rate1}) implies that $\mathcal{C} (h_0 P_S(i)) + R_B(i+1) \leq \mathcal{C} (P_S(i)) $, $i=1,\cdots,N$, which leads to
\begin{align} \label{NDC_rate2}
   \sum_{i=1}^k \mathcal{C}(P_S(i)) - \mathcal{C} (h_0 P_S(i)) \geq  \sum_{i=1}^k R_B(i+1), ~k=1,\cdots,N.
\end{align}
From (\ref{NDC_const2}) and (\ref{NDC_rate2}), we obtain
\begin{align}
\sum_{i=1}^k \mathcal{C} \left( h_0 P_S (i) \right) & +  \mathcal{C} \left( P_R (i+1) \right) \nonumber \\
 &~~\leq \sum_{i=1}^k \mathcal{C} \left( P_S (i) \right),~k=1,\cdots,N. \label{NDC_rate3}
\end{align}

Using (\ref{NDC_rate1}) and (\ref{NDC_rate3}), the average throughput for the NDC case is maximized by solving Problem (P2), which is non-convex due to the first constraint in (\ref{inf_buffer_DF2}), and thus difficult to solve at a first glance. We will derive the optimal solution for this problem based on a separate source and relay power allocation strategy in Section V. We conclude this section by the following two propositions regarding Problems (P1) and (P2).

\begin{Proposition} \label{compare_two_schemes}
The maximum value of Problem (P2) is no smaller than that of Problem (P1).
\end{Proposition}
\begin{proof}
By Proposition \ref{no_buffer_cons}, it follows that Problem (P1) has the same maximum value as Problem (P$1^*$) in (\ref{new2111})-(\ref{new2}). It is easy to see that any solution that satisfies (\ref{new2}) will also satisfy (\ref{inf_buffer_DF2}) for Problem (P2), but not vice versa. As such, the feasible set of Problem (P2) contains that of Problem (P$1^*$), which implies that the maximum value of Problem (P2) is no smaller than that of Problem (P$1^*$) or (P1). The proof is thus completed.
\end{proof}

\begin{Proposition} \label{no_buffer_source_energy}
For any optimal power profiles $P_S^*(i)$'s and $P_R^*(i+1)$'s for Problem (P1) or (P2), if $0 < h_0 < 1$, the constraint $B \sum_{i=1}^N P_S^*(i) \leq \sum_{i=1}^N E_S(i) $ must be satisfied with equality; if $h_0 = 0$, at least one of two constraints $B \sum_{i=1}^N P_S^*(i) \leq \sum_{i=1}^N E_S(i) $ and $B \sum_{i=1}^N P_R^*(i+1) \leq \sum_{i=1}^N E_R(i+1) $ must be satisfied with equality.
\end{Proposition}
\begin{proof}
Supposing that the above equalities do not hold, we can increase the source or relay transmission power without violating the energy constraints to further improve the throughput, which contradicts the fact that $P_S^*(i)$'s and $P_R^*(i+1)$'s are optimal solutions of Problem (P1) or (P2). Thus, this proposition is proved.
\end{proof}

\section{Optimal Solution for the DC Case}

In this section, we solve Problem (P1) for the DC case. We first present a monotonic property for the optimal power allocation in Problem (P1), upon which we then develop an efficient algorithm to solve this problem.

\subsection{Monotonic Power Allocation}

Since the optimal solution of Problem (P1) may not be unique (cf. Remark \ref{DC_opt_not_unique}), we are interested in finding one optimal solution for this problem that leads to the minimum power consumption at the source and relay. For such an optimal solution, we have the following monotonic property.

\begin{Proposition} \label{no_buffer_increasing_profile}
The optimal solution of Problem (P1), satisfying (\ref{no_buffer_cons_equa}) for the case of $0<h_0<1$ or (\ref{no_buffer_power_nodire}) for the case of $h_0=0$, is non-decreasing over $i$, i.e., $P_S^*(i) \leq P_S^*(i+1)$, $P_R^*(i+1) \leq P_R^*(i+2)$, $i=1,2,\cdots,N-1$.
\end{Proposition}
\begin{proof}
See Appendix \ref{proof_no_increase_profile}.
\end{proof}
With the above monotonic properties, we next address the optimal solution of Problem (P1) for the cases with and without the direct source-destination link, respectively.
\setcounter{equation}{29}

\subsection{The Case With Direct Link}

For the case of $0 < h_0 < 1$, we consider the Lagrangian of Problem (P1), which is given by (\ref{no_buffer_lag}), where $\mu_k,~\lambda_k,~\gamma_i$, and $\eta_{i+1}$ are the non-negative Lagrangian multipliers. By taking the derivative over $P_S(i)$ and $P_R(i+1)$, (\ref{deter_with_lag1})-(\ref{deter_with_lag2}) are obtained. Then, by letting $\frac{\partial \mathcal{L}}{ \partial P_S(i)}=0$ and $\frac{\partial \mathcal{L}}{ \partial P_R(i+1)}=0$, the optimal solutions for Problem (P1) are obtained as follows:
\begin{enumerate}
  \item Case I: if $ P_R^*(i+1) \geq \frac{(1-h_0)P_S^*(i)}{1+ h_0 P_S^*(i)} $,
\begin{align} \label{no_buffer_water1}
\left\{
  \begin{array}{l}
    P_S^*(i)  = \left( \frac{1}{4(N+1)\sum_{k=i}^{N} \mu_k} -1  \right)^+  \\
    P_R^*(i+1)  = \frac{(1-h_0)P_S^*(i)}{1+ h_0 P_S^*(i)}
  \end{array}
\right.;
\end{align}

  \item Case II: if $ P_R^*(i+1) \leq \frac{(1-h_0)P_S^*(i)}{1+ h_0 P_S^*(i)} $,
\begin{align} \label{no_buffer_water2}
\left\{
  \begin{array}{l}
    P_S^*(i)  = \left( \frac{1}{4(N+1)\sum_{k=i}^{N} \mu_k} - \frac{1}{h_0}  \right)^+  \\
    P_R^*(i+1)  = \left( \frac{1}{4(N+1)\sum_{k=i}^{N} \lambda_k} -1  \right)^+
  \end{array}
\right..
\end{align}
\end{enumerate}

\begin{Remark}
From the above expressions, it is observed that the source and relay power profiles need to be jointly optimized, since the achievable rate for the $i$-th source message is limited by the available source power in Case I, but by the available relay power in Case II.
\end{Remark}

From the KKT optimality conditions of Problem (P1), it follows that $\lambda_k$ and $\mu_k$ are strictly positive only when their corresponding relay and source energy constraints are satisfied with equality. Thus, it follows that the optimal source power can change the value from one block to another only when the harvested source energy is exhausted at the current block or there is a transition between the two values given by (\ref{no_buffer_water1}) and (\ref{no_buffer_water2}). The latter case is due to the fact that $h_0 <1$, and thus the source power values given by (\ref{no_buffer_water1}) and (\ref{no_buffer_water2}) are different even when the source energy constraint is not active. By further considering the result of Proposition \ref{no_buffer_increasing_profile}, we know that changing source power values from (\ref{no_buffer_water1}) to (\ref{no_buffer_water2}) is not possible, and only transitions from (\ref{no_buffer_water2}) to (\ref{no_buffer_water1}) can occur. Thus, we have the following proposition for the optimal source power allocation.

\begin{Proposition} \label{no_buffer_threecase}
Consider the optimal source power $P_S^*(i)$'s for Problem (P1), which satisfy Proposition \ref{no_buffer_increasing_profile}. For any two successive source energy exhausting blocks, $k_i$ and $k_j$ with $k_i < k_j$, i.e., the source energy constraints $ \sum_{i=1}^{k_i} P_S^*(i)  \leq \frac{1}{ B} \sum_{i=1}^{k_i} E_S(i)$ and $ \sum_{i=1}^{k_j} P_S^*(i)  \leq \frac{1}{ B} \sum_{i=1}^{k_j} E_S(i)$ are active, while the other constraints $ \sum_{i=1}^{j} P_S^*(i)  \leq \frac{1}{ B} \sum_{i=1}^{j} E_S(i)$, $j=k_i +1,\cdots,k_j-1$, are all inactive, the optimal source power values from the $\left( k_i+1 \right)$-th to the $k_j$-th blocks can only be one of the following three cases:
\begin{enumerate}
  \item Scenario I: $P_S^*(i)$, $i=k_i+1,\cdots,k_j$, are identical and given by (\ref{no_buffer_water1});
  \item Scenario II: $P_S^*(i)$, $i=k_i+1,\cdots,k_j$, are identical and given by (\ref{no_buffer_water2});
  \item Scenario III: There exists $k_0$ with $k_i <k_0<k_j$ such that $P_S^*(i)=P_0$, $i=k_i+1,...,k_0$, and $P_S^*(i)=P_0 + \frac{1}{h_0}-1$, $i=k_0+1,...k_j$, where $P_0$ is jointly determined by (\ref{no_buffer_water1}) and (\ref{no_buffer_water2}). Define the $k_0$-th block as the source power transition block for this scenario.
\end{enumerate}
\end{Proposition}

Based on Proposition \ref{no_buffer_threecase}, we know that if we could identify all blocks at which the source energy gets exhausted and furthermore all the scenarios corresponding to Proposition \ref{no_buffer_threecase}, the optimal source and relay power profiles for Problem (P1) can be obtained accordingly from (\ref{no_buffer_water1}) and (\ref{no_buffer_water2}). Thereby, we propose Algorithm \ref{no_buffer_algorithm} summarized in Table \ref{no_buffer_algorithm} to solve Problem (P1), whose optimality proof is given in Appendix \ref{proof_no_optimal}. The main procedure of this proposed algorithm is described as follows.

Starting from the first block, Algorithm \ref{no_buffer_algorithm} implements a forward searching for the optimal power allocation until the $N$-th block is reached. Suppose that the $(i-1)$-th block is an energy exhausting block for the source, and the optimal power allocation for the source and relay have been obtained up to this block, denoted by $P_S^*(n)$ and $P_R^*(n+1)$, respectively, $n=1,\cdots,i-1$. Then, starting from the $i$-th source and $(i+1)$-th relay transmission blocks, we first compute $i_{s,0}$ and $i_{r,0}$ (defined in (\ref{no_buffer_posi})), corresponding to the next possible source and relay energy exhausting blocks, respectively, and the source and relay power values $\widetilde{P}_S^{i,0}$ (from the $i$-th to $i_{s,0}$-th blocks) and $\widetilde{P}_R^{i+1,0}$ (from the $(i+1)$-th to $(i_{r,0}+1)$-th blocks), respectively, which are given as
\begin{align} \label{no_buffer_posi}
\left\{
  \begin{array}{l}
    i_{s,0} = \arg \min_{i \leq j \leq {N}} \left\{ \frac{ \sum_{k=i }^{j} E_S(k)}{ \left( j -i +1 \right)B }  \right\}, \\
    i_{r,0} = \arg \min_{i \leq j \leq {N}} \left\{ \frac{\widetilde{E}_R(i+1) +\sum_{k=i }^{j} E_R(k+1 )}{\left( j -i +1 \right)B  }  \right\},
  \end{array}
\right.
\end{align}
with $\widetilde{E}_R(i+1)$ denoting the relay energy left before the $(i+1)$-th relay transmission block, i.e., $\widetilde{E}_R(2) = 0$ and $\widetilde{E}_R(i+1) = \sum_{k=1}^{i-1}E_R(k+1) - B P_R^*(k+1)$, $i=2,\cdots,N$, and
\begin{align}
\widetilde{P}_S^{i,0}  & = \frac{  \sum_{k=i}^{i_{s,0}} E_S(k) }{\left( i_{s,0} -i +1 \right)B  }, \label{no_buffer_pswide1} \\
\widetilde{P}_R^{i+1,0} & = \frac{\widetilde{E}_R(i+1) + \sum_{k=i}^{i_{r,0}} E_R(k+1)}{ \left( i_{r,0} -i +1 \right) B  }. \label{no_buffer_pswide2}
\end{align}

Next, by comparing $\widetilde{P}_S^{i,0}$ and $\widetilde{P}_R^{i+1,0}$, we determine which scenario shown in Proposition \ref{no_buffer_threecase} should happen:

1) If $\widetilde{P}_R^{i+1,0}  \geq \frac{(1-h_0) \widetilde{P}_S^{i,0} }{1+ h_0 \widetilde{P}_S^{i,0} }$, it is claimed that Scenario I happens, and the optimal source and relay power values are given as
      \begin{align}
      \left\{
        \begin{array}{l}
          P_S^*(n)  = \widetilde{P}_S^{i,0} \\
          P_R^*(n+1)  = \frac{(1-h_0)\widetilde{P}_S^{i,0}}{1+ h_0 \widetilde{P}_S^{i,0}}
        \end{array}
      \right.
      ,~n=i,\cdots,i_{s,0}. \label{search_sceI_power}
      \end{align}
      Then, we set $i=i_{s,0}+1$, and continue the forward search.

2) If $ \widetilde{P}_R^{i+1,0}  < \frac{(1-h_0) \widetilde{P}_S^{i,0} }{1+ h_0 \widetilde{P}_S^{i,0} }$, Scenario II or III may happen. To determine whether Scenario III happens or not, we need to compute the index $k_0$, $1 < k_0 < N$, of the source power transition block defined in Scenario III of Proposition \ref{no_buffer_threecase}.

      Let the indices $i_{s,k}$ and $i_{r,p}$, $k,p \geq 1$, correspond to the source and relay energy exhausting blocks after $i_{s,k-1}$ and $i_{r,p-1}$, respectively, and $\widetilde{P}_S^{i,k}$ and $\widetilde{P}_R^{i+1,p}$ be the source and relay power values between the two successive energy exhausting blocks, i.e.,
        \begin{align}
        i_{s,k} & = \arg \min_{i_{s,k-1} < j \leq {N}} \left\{ \frac{  \sum_{k=i_{s,k-1}+1}^{j} E_S(k)}{\left( j -i_{s,k-1}   \right)B }  \right\}, \label{search_sceIII_power1} \\
        i_{r,p} & = \arg \min_{i_{r,p-1} < j \leq {N}} \left\{ \frac{  \sum_{k=i_{r,p-1}+1}^{j} E_R(k+1)}{\left( j -i_{r,p-1}  \right)B }  \right\}, \label{search_sceIII_power2} \\
        \widetilde{P}_S^{i,k}  & = \frac{ \sum_{k=i_{s,k-1}+1}^{i_{s,k}} E_S(k) }{ \left( i_{s,k} - i_{s,k-1}  \right)B  }, \label{search_sceIII_power31} \\
        \widetilde{P}_R^{i+1,p} & = \frac{ \sum_{k=i_{r,p-1}+1}^{i_{r,p}} E_R(k+1) }{\left( i_{r,p} - i_{r,p-1}  \right) B } . \label{search_sceIII_power32}
        \end{align}
        Then, we find $k_0>i$ such that
        \begin{align}
        \widetilde{P}_R(j+1) & < \frac{(1-h_0)\widetilde{P}_S(j)} {1+ h_0 \widetilde{P}_S(j)},~j=i,\cdots,k_0, \label{search_sceIII_k01} \\
        \widetilde{P}_R(k_0+2) &  > \frac{(1-h_0)\widetilde{P}_S(k_0+1)}{1+ h_0 \widetilde{P}_S(k_0+1)}, \label{search_sceIII_k02}
        \end{align}
        where $\widetilde{P}_S(j) = \widetilde{P}_S^{i,k}$ for $i_{s,k-1}<j \leq i_{s,k}$, and $\widetilde{P}_R(j+1) = \widetilde{P}_R^{i+1,p}$ for $i_{r,p-1}<j \leq i_{r,p}$, $k,p \geq 0$, assuming $i_{s,-1} = i_{r,-1} = i-1$.

        Suppose that $k_0$ is found in the above. Assume that the source power values from block $i$ to $k_0$ are $\widehat{P}_S^i$, and from block $k_0+1$ to the next source energy exhausting block denoted by $j_s$ are $\widehat{P}_S^i + \frac{1}{h_0} -1$, where $j_s$ is given by
        \begin{align} \label{search_sceIII_js}
        j_s = \arg \min_{k_0 +1 \leq j \leq {N}} \left\{ \frac{ \sum_{k=i}^{j} E_S(k) - (j- k_0 )(\frac{1}{h_0}-1)B  }{(j-i+1)B } \right\},
        \end{align}
        and $\widehat{P}_S^i$ is given as
        \begin{align} \label{search_sceIII_power4}
        \widehat{P}_S^i = \frac{ \sum_{k=i}^{j_s} E_S(k) - (j_s - k_0 ) (\frac{1}{h_0}-1) B } { (j_s - i +1)B }.
        \end{align}
        Then, we further check the following conditions inspired by Proposition \ref{no_buffer_increasing_profile} and Scenario III in Proposition \ref{no_buffer_threecase}:
        \begin{align}
         &\widetilde{P}_S^{i,0} \geq \widehat{P}_S^i \geq \max \left\{  \frac{ \widetilde{P}_R(k_0+1)} { 1- h_0 - h_0 \widetilde{P}_R(k_0+1) }, P_S^*(i-1) \right\}, \label{search_sceIII_power11} \\
         & \widehat{P}_S^i +  \frac{1}{h_0} - 1  \leq \frac{ \widetilde{P}_R(k_0+2)} { 1- h_0 - h_0 \widetilde{P}_R(k_0+2) }. \label{search_sceIII_power22}
        \end{align}
        If (\ref{search_sceIII_power11}) and (\ref{search_sceIII_power22}) are true, it is confirmed that Scenario III happens, and the optimal source and relay power values are given as
        \begin{align}
        P_S^*(n)& =\left\{
                   \begin{array}{ll}
                     \widehat{P}_S^i, & n=i,\cdots,k_0 \\
                     \widehat{P}_S^i + \frac{1}{h_0} -1, & n=k_0+1,\cdots, j_s
                   \end{array}
                 \right., \label{search_sceIII_power_final1} \\
        P_R^*(n+1) & = \left\{
                     \begin{array}{ll}
                       \widetilde{P}_R^{i+1,p}, & i_{r,p-1}< n \leq i_{r,p}, \\
                       & n\leq k_0 ,~p \geq 0 \\
                       \frac{(1-h_0) P_S^*(n)}{1 + h_0 P_S^*(n)}, & n=k_0+1,\cdots, j_s
                     \end{array}
                   \right.. \label{search_sceIII_power_final2}
        \end{align}
        Then, we set $i=j_s+1$, and continue the forward search.

If $k_0$ satisfying (\ref{search_sceIII_k01}) and (\ref{search_sceIII_k02}) cannot be found, or with such $k_0$ found in (\ref{search_sceIII_k01}) and (\ref{search_sceIII_k02}) but the conditions in (\ref{search_sceIII_power11}) and (\ref{search_sceIII_power22}) are not satisfied, we claim that Scenario III cannot happen and Scenario II must be true. In this case, the optimal source and relay power profiles are given as
      \begin{align}
      P_S^*(n) & = \widetilde{P}_S^{i,0},~P_R^*(n+1) = \widetilde{P}_R^{i+1,0}, \nonumber \\
       &~~~~~~~~~~~n=i,\cdots,\min(k_0,i_{s,0}), \label{search_sceII_power}
      \end{align}
      where $k_0 = \infty$ if no such $k_0$ satisfies (\ref{search_sceIII_k01}) and (\ref{search_sceIII_k02}). Then, we set $i=\min(k_0,i_{s,0})+1$, and continue the forward search.

\begin{table}[ht]
\begin{center}
\caption{Algorithm \ref{no_buffer_algorithm}: Compute the optimal solution of Problem (P1) for the case of $0 < h_0 <1$.}
\hrule
\vspace{0.3cm}
\begin{itemize}
\item Set $i=1$; while $i \leq N$, repeat
  \item Compute $i_{s,0}$ and $i_{r,0}$ by (\ref{no_buffer_posi}), and $\widetilde{P}_S^{i,0}$ and $\widetilde{P}_R^{i+1,0}$ by (\ref{no_buffer_pswide1}) and (\ref{no_buffer_pswide2}).
  \begin{enumerate}
    \item If $\widetilde{P}_R^{i+1,0}  \geq \frac{(1-h_0)\widetilde{P}_S^{i,0} }{1+ h_0 \widetilde{P}_S^{i,0} }$, compute $P_S^*(i)$ and $P_R^*(i+1)$ by (\ref{search_sceI_power}), and set $i=i_{s,0}+1$;
    \item If $\widetilde{P}_R^{i+1,0}  < \frac{(1-h_0)\widetilde{P}_S^{i,0} }{1+ h_0 \widetilde{P}_S^{i,0} }$, compute $i_{s,k}$, $i_{r,p}$, $\widetilde{P}_S^{i,k}$ and $\widetilde{P}_R^{i+1,p}$, $k,p \geq 1$ by (\ref{search_sceIII_power1})-(\ref{search_sceIII_power32}), and check
        \begin{enumerate}
          \item If there exists $k_0$ satisfying (\ref{search_sceIII_k01}) and (\ref{search_sceIII_k02}), compute $j_s$ and $\widehat{P}_S^i$ by (\ref{search_sceIII_js}) and (\ref{search_sceIII_power4}), respectively, and check
          \begin{enumerate}
            \item If $\widehat{P}_S^i$ satisfies (\ref{search_sceIII_power11}) and (\ref{search_sceIII_power22}), compute $P_S^*(i)$ and $P_R^*(i+1)$ by (\ref{search_sceIII_power_final1}) and (\ref{search_sceIII_power_final2}); set $i=j_s+1$.
            \item else compute $P_S^*(i)$ and $P_R^*(i+1)$ by (\ref{search_sceII_power}); set $i=\min\{k_0, i_{s,0}\}+1$.
          \end{enumerate}

          \item else compute $P_S^*(i)$ and $P_R^*(i+1)$ by (\ref{search_sceII_power}); set $i=i_{s,0}+1$.
        \end{enumerate}
  \end{enumerate}

  \item Algorithm ends.
\end{itemize}
\vspace{0.2cm} \hrule \label{no_buffer_algorithm} \end{center}
\end{table}

\begin{Remark}
According to the proof given in Appendix \ref{proof_no_optimal}, it follows that for the case of $0< h_0 <1$, the optimal source power solution of Problem (P1) obtained using Algorithm \ref{no_buffer_algorithm} is unique, while this is not necessarily true for the obtained optimal relay power solution according to Proposition \ref{no_buffer_cons}. However, the obtained relay power solution achieves the minimum energy consumption at the relay, since (\ref{no_buffer_cons_equa}) is satisfied for each $i$, $i=1,\cdots,N$.
\end{Remark}

\subsection{The Case Without Direct Link}

Similar to the case with direct link, we obtain the following optimal power solutions for Problem (P1) in the case of $h_0=0$:
\begin{enumerate}
  \item If $ P_R^*(i+1) \geq P_S^*(i) $,
\begin{align}\label{no_direct_1}
\left\{
  \begin{array}{l}
    P_S^*(i)  = \left( \frac{1}{4(N+1)\sum_{k=i}^{N} \mu_k} -1  \right)^+  \\
    P_R^*(i+1)  = P_S^*(i)
  \end{array}
\right. ;
\end{align}

  \item If $ P_R^*(i+1) \leq P_S^*(i) $,
\begin{align} \label{no_direct_2}
\left\{
  \begin{array}{l}
    P_S^*(i)  = P_R^*(i+1) \\
    P_R^*(i+1)  = \left( \frac{1}{4(N+1)\sum_{k=i}^{N} \lambda_k} -1  \right)^+
  \end{array}
\right. .
\end{align}
\end{enumerate}
It is worth noting that to obtain the optimal power profile with the minimum energy consumption, we set the source and relay power levels the same, while in general this is not necessary since the optimal source and relay power profiles may not be unique.

From the above solutions, it is observed that: 1) The source power $P_S^*(i)$ at the $i$-th block and relay power $P_R^*(i+1)$ at the $(i+1)$-th block should be identical, $i=1,\cdots,N$; 2) Since $\lambda_k$ and $\mu_k$ are strictly positive only when their corresponding energy constraints are satisfied with equality, it follows that the source/relay power changes value only when the harvested energy at either the source or the relay is exhausted.

Based on the above observations, Algorithm \ref{no_buffer_algorithm} for the case with direct link can be simplified to obtain the optimal source and relay power allocation for Problem (P1) in the case without direct link. We denote this algorithm as Algorithm \ref{no_buff_algo}, which is summarized in Table \ref{no_buff_algo}. Since the optimality proof of Algorithm II is similar to that of Algorithm I, it is omitted here for brevity.

\begin{table}[ht]
\begin{center}
\caption{Algorithm \ref{no_buff_algo}: Compute the optimal solution of Problem (P1)  for the case of $h_0 = 0$.}
\hrule
\vspace{0.3cm}
\begin{itemize}
\item Set $i=1$; when $i \leq N$, repeat
  \item Compute $\widetilde{P}_S^i$, $\widetilde{P}_R^{i+1}$, $i_s$, and $i_r$ as follows.
      \begin{align*}
      i_s & = \arg \min_{i \leq j \leq {N}} \left\{ \frac{\widetilde{E}_S(i) + \sum_{k=i}^{j} E_S(k)}{\left( j -i +1 \right)B }  \right\} , \\
      i_r & = \arg \min_{i \leq j \leq {N}} \left\{ \frac{\widetilde{E}_R(i+1) +\sum_{k=i}^{j} E_R(k+1)}{\left( j -i +1 \right)B  }  \right\} ,\\
      \widetilde{P}_S^i & = \frac{ \widetilde{E}_S(i) + \sum_{k=i}^{i_{s}} E_S(k) }{\left( i_{s} -i +1 \right) B  } , \\
      \widetilde{P}_R^{i+1} & = \frac{\widetilde{E}_R(i+1) + \sum_{k=i}^{i_{r}} E_R(k+1)}{ \left( i_{r} -i +1 \right) B },
      \end{align*}
      where $\widetilde{E}_S(1)= \widetilde{E}_R(2)=0$, $\widetilde{E}_S(i) = \sum_{k=1}^{i-1}E_S(k) - B P_S^*(k)$ and $\widetilde{E}_R(i+1) = \sum_{k=1}^{i-1}E_R(k+1) - B P_R^*(k+1)$, $i=2,\cdots,N$, respectively. Then, the optimal source and relay power profiles are given as
      \begin{enumerate}
        \item If $\widetilde{P}_S^i \geq \widetilde{P}_R^{i+1}$, set $P_S^*(k) = P_R^*(k+1)=\widetilde{P}_R^{i+1}$ for $k=i,\cdots,i_r$, and $i=i_r+1$;
        \item If $\widetilde{P}_S^i < \widetilde{P}_R^{i+1}$, set $P_S^*(k) = P_R^*(k+1)=\widetilde{P}_S^i$ for $k=i,\cdots,i_s$, and $i=i_s+1$.
      \end{enumerate}
  \item Algorithm ends.
\end{itemize}
\vspace{0.2cm} \hrule \label{no_buff_algo} \end{center}
\end{table}

\begin{figure}[!t]
\centering
\includegraphics[width=.9 \linewidth]{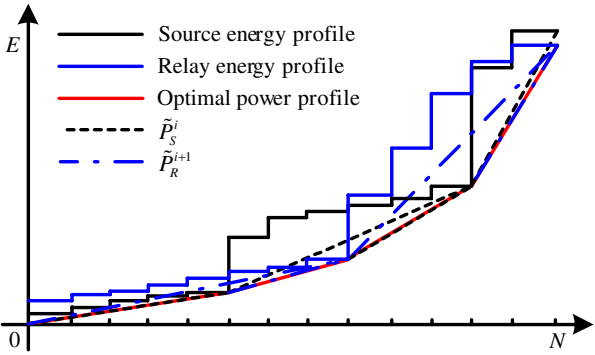}
\caption{An example for the optimal source and relay power allocation for the DC case with $h_0=0$. Note that the slope of each dashed line corresponds to $\widetilde{P}_S^i$ or $\widetilde{P}_R^{i+1}$ obtained using Algorithm \ref{no_buff_algo}.}
\label{no_direct}
\end{figure}

\begin{Remark}
The optimal solution for Problem (P1) obtained using Algorithm \ref{no_buff_algo} in the case of $h_0=0$ is illustrated in Fig. \ref{no_direct}, where we see that the optimal source/relay power profile corresponds to the shortest path from the origin to the highest possible sum-energy point at the end of the $N$-block transmission under two  stair-like source and relay energy constraints. It is worth noting that for the case of $h_0=0$, our relay channel model degrades to a cascade of two AWGN point-to-point channels that were studied in \cite{yang} with individual EH constraints. Thus, at each block, either the source-relay or the relay-destination link can be the performance bottleneck.
\end{Remark}

\section{Optimal Solution for the NDC Case}

In this section, we solve Problem (P2) for the NDC case. We first prove that a separation principle for the source and relay power allocation problem holds, upon which Problem (P2) can be solved by a two-stage strategy: First obtain the optimal source power allocation by ignoring the relay, and then optimize the relay power allocation with the obtained source power solution. Since this separation principle applies to both cases with and without the direct link, we address these two cases with a unified analysis as follows.

\subsection{Optimal Source Power Allocation}

First, we consider the following source power allocation problem by ignoring the relay:
\begin{align}
(\text{P3})&~~\max_{ P_S(i) \geq 0, \forall i} ~~ \sum_{i=1}^N \mathcal{C} \left( h P_S (i) \right)  \\
\text{s. t.}&~~ \sum_{i=1}^k P_S(i)  \leq   \frac{1}{B} \sum_{i=1}^k E_S(i),~k=1,2,\cdots,N, \label{no_constr_source_cons}
\end{align}
where $h$ is a constant with $0 < h \leq 1$. Problem (P3) has been solved in \cite{yang}, for which the algorithm (denoted by Algorithm \ref{with_source_profile}) is summarized in Table \ref{with_source_profile} for the sake of completeness. Note that the optimal source power profile $P_S^*(i)$'s of Problem (P3) are non-decreasing over $i$ \cite{yang}.

\begin{table}[ht]
\begin{center}
\caption{Algorithm \ref{with_source_profile}: Compute the optimal solution for Problem (P3).}
\hrule
\vspace{0.3cm}
\begin{itemize}
\item Initialize $i=1$; while $i \leq N$, repeat
  \item Compute
\begin{align*}
i_s = \arg \min_{i \leq j \leq {N}} \left\{ \frac{\sum_{k=i}^{j} E_S(k)}{\left( j -i +1 \right) B  }  \right\},~P_S^i = \frac{ \sum_{k=i}^{i_s} E_S(k) }{ \left( i_s-i +1 \right)B  }.
\end{align*}
The optimal source power profile is given as $P_S^*(n)= P_S^i,~n=i,\cdots,i_s$; set $i=i_s+1$.
  \item Algorithm ends.
\end{itemize}
\vspace{0.2cm} \hrule  \label{with_source_profile} \end{center}
\end{table}

It is easy to see that the optimal source power profile of Problem (P3) maximizes the average throughput of both the source-relay and source-destination links. Moreover, since for the NDC case, the relay can store the binning indices of the decoded source messages with arbitrary delay before forwarding them to the destination with best effort transmissions, the relay power profile intuitively should have no effect on the optimal source power profile. This conjecture is affirmed by the following proposition.

\begin{Proposition} \label{with_source_profile_opt}
For the NDC case, the optimal source power solution for Problem (P3) is also globally optimal for Problem (P2).
\end{Proposition}
\begin{proof}
See Appendix \ref{with_source_opt}.
\end{proof}

This proposition implies that the separation principle for the source and relay power allocation problems is optimal for Problem (P2). Thus, even though Problem (P2) is non-convex, we can still find its globally optimal solution efficiently.

\begin{Remark} \label{rem_source_NDC}
It is shown in Appendix \ref{with_source_opt} that for the case of $0 < h_0 <1$, the optimal source power profile $P_S^*(i)$'s of Problem (P2) given by Algorithm \ref{with_source_profile} is unique. However, for the case of $h_0=0$, this result is not true in general, since the source energy may not necessarily be exhausted at the end of each $N$-block transmission for Problem (P2) (cf. Proposition \ref{no_buffer_source_energy}).
\end{Remark}

\subsection{Optimal Relay Power Allocation}

With the optimal source power profile $P_S^*(i)$'s obtained using Algorithm \ref{with_source_profile}, the optimal relay power profile can be determined by the following problem:
\begin{align}
(\text{P4})&~~\max_{P_R(i+1)\geq 0, \forall i} ~~  \sum_{i=1}^N \mathcal{C} \left( P_R (i+1) \right) \label{inf_buffer_relay_power1} \\
\text{s. t.} &~~ \sum_{i=1}^k \mathcal{C} \left( P_R (i+1) \right) \leq \sum_{i=1}^k \mathcal{C} \left( P_S^* (i) \right) \nonumber \\
&~~- \sum_{i=1}^k \mathcal{C} \left( h_0 P_S^* (i) \right), ~k=1,\cdots,N,~\text{and},(\ref{source_relay_energy-const2}). \label{inf_buffer_relay_power2}
\end{align}
This problem is non-convex due to the constraints in (\ref{inf_buffer_relay_power2}). However, by letting $r(i+1) =  \mathcal{C} \left( P_R (i+1) \right)$, Problem (P4) can be rewritten as
\begin{align}
(\text{P5})&~~\max_{r(i+1)\geq 0, \forall i} ~~  \sum_{i=1}^N r(i+1) \label{inf_buffer_relay_power_new1} \\
\text{s. t.} &~~ \sum_{i=1}^k r(i+1)  \leq   \sum_{i=1}^k \mathcal{C} \left( P_S^* (i) \right) - \mathcal{C} \left( h_0 P_S^* (i) \right), \label{inf_buffer_relay_power_new2} \\
&~~ \sum_{i=1}^k \left( 2^{2 r(i+1)} -1 \right) \leq \frac{1}{B}\sum_{i=1}^k E_R(i+1), \nonumber \\
&~~~~~~~~~~~~~~~~~~~~~~~~~~~~~~~~~~~~k=1,\cdots,N.  \label{inf_buffer_relay_power_new3}
\end{align}
It can be shown that Problem (P5) is convex over $r(i+1)$'s. By the KKT optimality conditions, we obtain the optimal solution for Problem (P5) as
\begin{align} \label{with_buffer_relay}
r^*(i+1) = \left( \frac{1}{2} \log \frac{1 - \sum_{k=i}^{N} \lambda_k }{2 \ln 2 \cdot \sum_{k=i}^{N} \gamma_k } \right)^+,~i=1,\cdots,N
\end{align}
where $\lambda_k$ and $\gamma_k$ are the non-negative Lagrangian multipliers corresponding to the $k$-th constraint in (\ref{inf_buffer_relay_power_new2}) and (\ref{inf_buffer_relay_power_new3}), respectively. Problem (P5) can be solved by a forward search algorithm, denoted by Algorithm \ref{inf_buffer_relay_rate} in Table \ref{inf_buffer_relay_rate}, for which the optimality proof is similar to that of Algorithm \ref{no_buff_algo}, and is thus omitted.

\begin{table}[ht]
\begin{center}
\caption{Algorithm \ref{inf_buffer_relay_rate}: Compute the optimal solution for Problem (P5).}
\hrule
\vspace{0.3cm}
\begin{itemize}
\item Initialize $i=1$; while $i \leq N$, repeat
  \item Compute
\begin{align*}
& i_1 = \arg \min_{i \leq j \leq {N}} \left\{ \frac{\widetilde{C}_{i} + \sum_{k=i}^{j}  \mathcal{C} \left( P_S^* (k) \right) -  \mathcal{C} \left( h_0 P_S^* (k) \right)}{\left( j -i +1 \right)B  }  \right\}, \\
& i_2= \arg \min_{i \leq j \leq {N}} \left\{ \frac{\widetilde{E}_{i+1} + \sum_{k=i}^{j} E_S(k+1)}{\left( j -i +1 \right) B  }  \right\}, \\
& \widetilde{r}_1 = \frac{ \widetilde{C}_{i} + \sum_{k=i}^{i_1}  \mathcal{C} \left( P_S^* (k) \right) -  \mathcal{C} \left( h_0 P_S^* (k) \right) }{ \left( i_1 - i +1 \right) B }, \\
&\widetilde{r}_2= \mathcal{C} \left(  \frac{\widetilde{E}_{i+1} + \sum_{k=i}^{i_2} E_S(k+1)}{\left( i_2 - i +1 \right) B }  \right),
\end{align*}
where $\widetilde{C}_1 = \widetilde{E}_2=0$, $\widetilde{C}_i= \sum_{k=1}^{i-1}  \mathcal{C} \left( P_S^* (k) \right) -  \mathcal{C} \left( h_0 P_S^* (k) \right) - r^*(k) $, and $\widetilde{E}_{i+1}= \sum_{k=1}^{i-1} \left( E_S(k+1) -  2^{2 r^*(k+1)} -1 \right)$, $i=2,\cdots,N$. Let $j_0 =\arg \min_{j=1,2}\left\{ \widetilde{r}_j \right\}$. Set $r^*(j+1) = \widetilde{r}_{j_0}, ~j=i,\cdots,i_{j_0}$, and $i=i_{j_0}+1$.
  \item Algorithm ends.
\end{itemize}
\vspace{0.2cm} \hrule \label{inf_buffer_relay_rate} \end{center}
\end{table}

\begin{Remark}
It is worth noting that from (\ref{with_buffer_relay}), we observe that the optimal relay transmission rate $r^*(i+1)$ is non-decreasing over $i$, and strictly increases when any one of the constraints (\ref{inf_buffer_relay_power_new2}) and (\ref{inf_buffer_relay_power_new3}) is satisfied with equality. Thus, the optimal relay power profile $P_R^*(i+1)$'s of Problem (P2) with $P_R^*(i+1)=2^{2r^*(i+1)}-1,~i=1,\cdots,N$, are also non-decreasing over $i$. Furthermore, it can be shown that the obtained $P_R^*(i+1)$'s achieve the minimum energy consumption at the relay.
\end{Remark}

Based on the analysis in the above two subsections, we obtain the following proposition for the monotonic properties for $P_S^*(i)$'s and $P_R^*(i+1)$'s for the NDC case, similar to the DC case (cf. Proposition \ref{no_buffer_increasing_profile}).

\begin{Proposition} \label{pro_mono_NDC}
For Problem (P2), there exist optimal source power solution $P_S^*(i)$'s and relay power solution $P_R^*(i+1)$'s that are non-decreasing over $i,~i=1,\cdots,N$.
\end{Proposition}

\begin{Remark}
It is worth noting that problems of similar structures as (P4) have been solved independently in prior works [3], where the results are based on generalizing the solution for the two-epoch case. Nevertheless, in this paper, we use some optimization tools to transfer Problem (P4) into Problem (P5), which is convex and much easier to be solved.
\end{Remark}

\subsection{Optimal Rate Scheduling}

It has been shown in Section IV that for the DC case, the relay binning rate $R_B(i+1)$ for each source message can be directly computed by (\ref{DC_binning}) given the optimal source and relay power profiles. However, for the NDC case, it requires additional effort to obtain $R_B(i+1)$, since we need to consider two sets of different constraints in (\ref{NDC_const2}) and (\ref{NDC_rate2}) (as will be shown next). Suppose that $R_B(i+1)$ is obtained with the optimal source and relay power profiles $P_S^*(i)$'s and $P_R^*(i+1)$'s, the source transmission rates $R(i)$'s can be determined from (\ref{NDC_rate1}). Then, with the obtained $R(i)$'s, we can update the obtained source power profile to achieve the same maximum throughput but with the minimum energy consumption.

\begin{figure}[!t]
\centering
\includegraphics[width=.9 \linewidth]{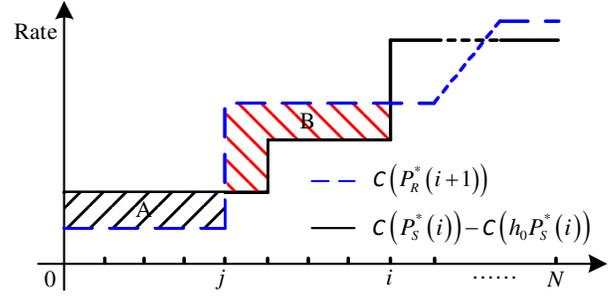}
\caption{Illustration of the rate allocation for $R_B(i+1)$'s in the NDC case.}
\label{rate_allo}
\end{figure}

To compute $R_B(i+1)$'s, the following observations are first drawn. If $\mathcal{C} (P_R^*(i+1)) > \mathcal{C} (P_S^*(i)) - \mathcal{C} ( h_0 P_S^*(i))$, $\forall i \in \{1,\cdots,N\}$, the relay should transmit not only the binning index of the $i$-th source message at the $(i+1)$-th block, but also those of source messages $1 \leq j <i$. Moreover, due to the constraint in (\ref{inf_buffer_DF2}), it follows that if $\mathcal{C} (P_R^*(i+1)) > \mathcal{C} (P_S^*(i)) - \mathcal{C} ( h_0 P_S^*(i))$, $\forall i$, there must exist $j$ with $1 \leq j < i$, such that $\mathcal{C} (P_R^*(j+1)) < \mathcal{C} (P_S^*(j)) - \mathcal{C} ( h_0 P_S^*(j))$. The above observations imply that to obtain $R_B(i+1)$'s, we need to find all $i$'s with $\mathcal{C} (P_R^*(i+1)) > \mathcal{C} (P_S^*(i)) - \mathcal{C} ( h_0 P_S^*(i))$, and then use their surplus rates to transmit the binning indices of source messages $j \leq i$. In Fig. \ref{rate_allo}, we show an example for the relationship between $\mathcal{C} (P_R^*(i+1))$ and $\mathcal{C} (P_S^*(i)) - \mathcal{C} ( h_0 P_S^*(i))$, where Case A stands for the case of $\mathcal{C} (P_R^*(i+1)) < \mathcal{C} (P_S^*(i)) - \mathcal{C} ( h_0 P_S^*(i))$ and Case B for the case with a reversed inequality. Then, a geometric interpretation for the algorithm of computing $R_B(i+1)$'s is given as follows: Use the surplus rate area in Case B to fill the deficient rate area in Case A. Obviously, if the area to be filled in Case A is larger than the filling area in Case B, the corresponding values for $R_B(i+1)$'s are not unique.

Thus, we now develop a backward search algorithm (Algorithm \ref{binning_rate_compute}) that is summarized in Table \ref{binning_rate_compute}, to obtain one of the feasible solutions for $R_B(i+1)$'s. The main procedure of this algorithm is described as follows. First, $R_B(i+1)$'s are initialized as the minimum values between $\mathcal{C} (P_R^*(i+1)) $ and $ \mathcal{C} (P_S^*(i)) - \mathcal{C} ( h_0 P_S^*(i))$ for all $i$'s, and a parameter $t$ (which, e.g., corresponds to the filling area in Case B of Fig. 3) is set to be $0$. The algorithm then searches the values for $R_B(i+1)$'s in a backward way from $i=N$ to 1. For any $i$-th block, the algorithm computes $temp = \mathcal{C} (P_R^*(i+1)) - \left( \mathcal{C} (P_S^*(i)) - \mathcal{C} ( h_0 P_S^*(i)) \right)$. Then, if $temp>0$, Case B occurs, and $temp$ is added to $t$; if $temp<0$, Case A occurs, and $R_B(i+1)$ is increased by $\min( -temp, t )$, and this amount is then subtracted from $t$.

\begin{table}[ht]
\begin{center}
\caption{Algorithm \ref{binning_rate_compute}: Compute $R_B(i+1)$'s for the NDC case.}
\hrule
\vspace{0.3cm}
\begin{itemize}
 \item Given $P_S^*(i)$'s by Algorithm \ref{with_source_profile} and $P_R^*(i+1)$'s by Algorithm \ref{inf_buffer_relay_rate}; initialize $R_B(i+1)= \min \left\{ \mathcal{C} \left( P_R^*(i+1) \right), \mathcal{C} \left( P_S^*(i) \right) - \mathcal{C} \left( h_0 P_S^*(i) \right) \right\}$, $i=1,\cdots,N$, and $t=0$;
  \item From $i=N$ to $1$, compute $temp=\mathcal{C} \left( P_R^*(i+1) \right) - \mathcal{C} (P_S^*(i)) + \mathcal{C} ( h_0 P_S^*(i))$
  \begin{enumerate}
    \item If $temp>0$, set $t = t+temp $;
    \item If $temp<0$, set $R_B(i) = R_B(i) + \min(-temp,t) $ and $t = (t+temp)^+ $. \label{binning_rate_increase}
  \end{enumerate}
  \item Algorithm ends.
\end{itemize}
\vspace{0.2cm} \hrule \label{binning_rate_compute} \end{center}
\end{table}

\begin{Remark}
It is recalled in Remark \ref{rem_source_NDC} that in the case of $h_0=0$, the optimal source power profile $P_S^*(i)$'s given by Algorithm \ref{with_source_profile} may not achieve the minimum energy consumption. In this case, from (\ref{NDC_rate1}), it follows that $R(i) = R_B(i+1)$, $i=1,\cdots,N$, with $R_B(i+1)$'s obtained by using Algorithm \ref{binning_rate_compute}. In order to achieve the minimum energy consumption at the source, the optimal source power solution of Problem (P2) can be updated as $P_S^*(i) = 2^{2R_B(i+1)}-1$. It is worth noting that the above obtained source power profile is still non-decreasing over $i$ (thus in accordance with Proposition \ref{pro_mono_NDC}), since $R_B(i+1)$'s obtained using Algorithm \ref{binning_rate_compute} are non-decreasing over $i$.
\end{Remark}

\subsection{Throughput Comparison: DC vs. NDC}

As shown by Proposition \ref{compare_two_schemes}, the throughput of the NDC case is no smaller than that of the DC case. To further compare these two cases, the following proposition shows when the NDC case is strictly better than the DC case in terms of the average throughput.

\begin{Proposition} \label{compare_two}
The average throughput of the NDC case is strictly larger than that of the DC case if and only if there exists $i \in \{1 ,\cdots, N\}$ such that $\mathcal{C} (P_R^*(i+1)) > \mathcal{C} \left( P_S^*(i) \right) - \mathcal{C} \left( h_0 P_S^*(i) \right)$, where $ P_S^*(i)$'s and $\mathcal{C} (P_R^*(i+1))$'s are the optimal solutions to Problems (P3) and (P5) obtained by Algorithms \ref{with_source_profile} and \ref{inf_buffer_relay_rate}, respectively.
\end{Proposition}
\begin{proof}
See Appendix \ref{proof_compare}.
\end{proof}

\section{Numerical Results}

In this section, we present some numerical results to validate our theoretical results. For the purpose of exposition, we assume a periodic energy profile model for some predictable EH sources. Specifically, the source and relay energy profiles are given as
\begin{align*}
E_S(i) & = A_S \sin \left(\frac{i-1}{N} 2\pi  + \frac{\pi}{2} \right) + A_S, \\
E_R(i+1) & = A_R \sin \left(\frac{i-1}{N} 2 \pi + \theta \right) + A_R ,~ 1 \leq i \leq N,
\end{align*}
respectively, where $A_S,A_R>0$ are the amplitudes of the sinusoidal energy profiles at the source and relay, respectively, and $\theta$ is the phase shift between these two energy profiles. Here, we choose $B=100$, $N=40$, $\theta= \frac{5}{4} \pi$, and $A_S= A_R=200$.

We compare our proposed algorithms with a greedy power allocation strategy. Here, we adopt a non-trivial greedy algorithm by assuming that both the source and relay know the harvested energy amounts up to the current block prior to transmissions. The transmission rate of the $i$-th source message in the greedy algorithm is given as
\begin{align}
R_G(i) & = \min \left\{ \mathcal{C} \left( \frac{\widetilde{E}_S(i)}{B}  \right), \right. \nonumber \\
& ~~~~~~~\left. \mathcal{C} \left( h_0 \frac{\widetilde{E}_S(i)}{B} \right) + \mathcal{C} \left(\frac{\widetilde{E}_R(i+1)}{B}  \right)  \right\},
\end{align}
where $\widetilde{E}_S(i) = \sum_{k=1}^{i} E_S(k) -   B \sum_{k=1}^{i-1} P_S(k)$, and $\widetilde{E}_R(i+1) = \sum_{k=1}^{i} E_R(k+1) -  B  \sum_{k=1}^{i-1}P_R(k+1)$, with $P_S(i) = 2^{2R_G(i)} -1$ and $P_R(i+1) = 2^{2\left( R_G(i) - \mathcal{C} \left( h_0 \frac{\widetilde{E}_S(i)}{B} \right) \right)} -1$, $i=1,\cdots,N$. Note that in the above greedy algorithm, both the source and relay consume as much available power as possible at two successive blocks to maximize the instantaneous throughput, thus achieving the minimum delay as for the DC case.

\begin{figure}[!t]
\centering
\includegraphics[width=.9 \linewidth]{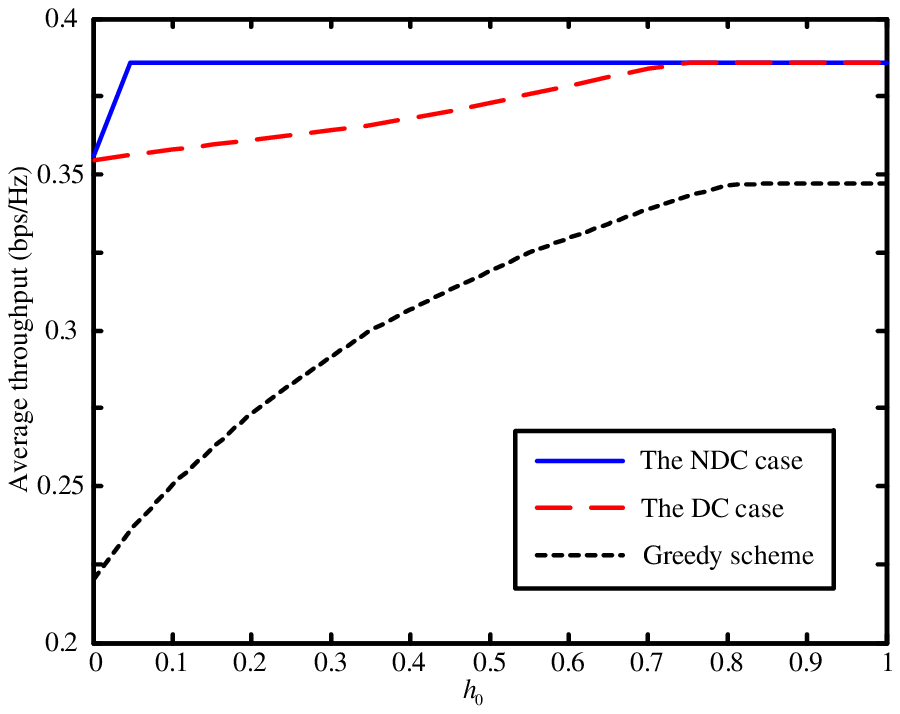}
\caption{Throughput comparison of various power allocation schemes for the relay channel with energy harvesting constraints.}
\label{diff_h0}
\end{figure}

In Fig. \ref{diff_h0}, we show the average throughputs versus the
direct link channel gain $h_0$ for the proposed power allocation
algorithms and the greedy algorithm. It is observed that as the
direct link becomes stronger, i.e., $h_0$ increases, there is a
throughput limit of 0.387 bps/Hz. For the NDC case, this throughput
limit is achieved even for very small $h_0$ around 0.05. In
contrast, for the DC case, the throughput increases almost linearly
and achieves the throughput limit when $h_0$ exceeds 0.75. Thus, the
throughput gain of NDC over DC cases by relaxing the decoding delay
is present only in the regime of small direct link gains, thanks to
the exploitation of energy diversity in cooperative communication.
Also note that with small $h_0$, the condition in Proposition
\ref{compare_two} is more likely to be satisfied. For the greedy
algorithm, it is observed that the throughput loss can be large,
especially when $h_0$ is small, as compared to the proposed
algorithm for the DC case.

\section{Concluding Remarks}

In this paper, we studied the throughput maximization problem for the orthogonal relay channel with EH source and relay nodes, assuming a deterministic EH model. For both the cases with and without delay constraints at the destination, we examined the structures of the optimal source and relay power profiles over time, and developed algorithms to efficiently compute these optimal power profiles. In addition, a new interesting energy diversity phenomenon was
explored in EH-powered wireless cooperative communication with delay
tolerance. We now conclude the paper by highlighting some important aspects unaddressed yet and thus worth being investigated in the future work as follows:
\begin{enumerate}
  \item This paper considers the deterministic EH model, while in practice, many EH sources should be modeled as random processes, e.g., the block-Markov model \cite{rui}. The study of the throughput maximization problem for the Gaussian relay channel under random EH models motivated by the results obtained here will thus be practically more appealing.
  \item To simplify the analysis, it is assumed in this paper that the energy storage capacity is infinite, which may not be true in practice \cite{yener}. Thus, considering the case with finite energy storage will be an interesting extension of this paper.
  \item In this paper, the DF relaying technique is adopted to design the optimal transmission for the orthogonal relay channel. Extensions to more general relay channel models and/or other relaying techniques, such as amplify and forward (AF) as well as CF, will be also interesting.
\end{enumerate}

\appendices

\section{Proof of Proposition \ref{no_buffer_increasing_profile}} \label{proof_no_increase_profile}

Denote the optimal source and relay power profiles of Problem (P1) as $P_S^*(i)$'s and $P_R^*(i+1)$'s, respectively. For any consecutive power pair, consider the following three cases:

1) $P_S^*(i)>P_S^*(i+1)$ and $P_R^*(i+1)>P_R^*(i+2)$: Define a new power allocation profile as $\widetilde{P}_S(i) = \widetilde{P}_S(i+1) = \frac{P_S^*(i) + P_S^*(i+1)}{2}$ and $\widetilde{P}_R(i+1) = \widetilde{P}_R(i+2) = \frac{P_R^*(i+1) + P_R^*(i+2)}{2}$. It is easy to check that the new power allocation profile still satisfies the energy constraint in (\ref{no_buffer_opt2}); moreover, since the objective function in (\ref{no_buffer_opt}) is concave, it follows that the new allocation leads to a larger sum rate over these two blocks. Thus, this case cannot happen.

2) $P_S^*(i)>P_S^*(i+1)$ and $P_R^*(i+1) \leq P_R^*(i+2)$: First, we prove that the $i$-th constraint in (\ref{no_buffer_cons_equa}) is not achieved with equality by contradiction as follows. Assuming that its equality is achieved, it is then observed that $P_R^*(i+1) \leq P_R^*(i+2)$ is contradicted with the following results:
      \begin{align}
      \log \left( 1+ P_R^*(i+1) \right) &  = \log \left( \frac{ 1 +P_S^*(i)}{ 1 + h_0 P_S^*(i)} \right) \nonumber \\
      & > \log \left( \frac{ 1 +P_S^*(i+1)}{ 1 + h_0 P_S^*(i+1)} \right) \label{no_buffer_incre_case2_equa1}\\
      & \geq   \log \left( 1+ P_R^*(i+2) \right), \label{no_buffer_incre_case2_equa2}
      \end{align}
  where (\ref{no_buffer_incre_case2_equa1}) is due to the fact that $\log \left( \frac{ 1 + x}{ 1 + h_0 x} \right)$ is strictly increasing over $x>0$ for any fixed $h_0$ with $0 \leq h_0 <1$, and (\ref{no_buffer_incre_case2_equa2}) is by Proposition \ref{no_buffer_cons}.

  Thus, there must exist $\delta$, $0 < \delta \leq \frac{P_S^*(i) - P_S^*(i+1)}{2}$, such that
  \begin{align} \label{no_buffer_incre_case2_equa3}
  \mathcal{C} \left(P_S^*(i) - \delta \right) \geq \mathcal{C} \left(h_0 \left( P_S^*(i)- \delta \right) \right) + \mathcal{C} \left( P_R^*(i+1) \right).
  \end{align}
  Define a new power allocation profile as $\widetilde{P}_S(i) = P_S^*(i) - \delta $, $\widetilde{P}_S(i+1) = P_S^*(i+1) +\delta$, $\widetilde{P}_R(i+1) = P_R^*(i+1)$, and $\widetilde{P}_R(i+2) = P_R^*(i+2)$. It is easy to check that the new allocation satisfies the energy constraints in (\ref{no_buffer_opt2}), and yields a larger sum rate over these two blocks, i.e.,
  \begin{align}
&~~ ~~ \widetilde{R}(i) + \widetilde{R}(i+1) \nonumber \\
& \geq \mathcal{C} \left(h_0 \widetilde{P}_S(i) \right) + \mathcal{C} \left( P_R^*(i+1) \right) + \mathcal{C} \left(h_0 \widetilde{P}_S(i+1) \right) \nonumber \\
&~~~~~~~~~~~~~~~~~ + \mathcal{C} \left( P_R^*(i+2) \right) \label{no_buffer_incre_case2_equa4} \\
& > \mathcal{C} \left(h_0 P_S^*(i) \right) + \mathcal{C} \left( P_R^*(i+1) \right) + \mathcal{C} \left(h_0 P_S^*(i+1) \right) \nonumber \\
&~~~~~~~~~~~~~~~~~ + \mathcal{C} \left( P_R^*(i+2) \right) \label{no_buffer_incre_case2_equa5} \\
& = R^*(i) +  R^*(i+1), \nonumber
  \end{align}
where (\ref{no_buffer_incre_case2_equa4}) is due to (\ref{no_buffer_incre_case2_equa3}) and the following fact: Since $\mathcal{C}(x + \delta)  - \mathcal{C} \left( x \right) \geq \mathcal{C} \left(h_0 (x + \delta) \right) - \mathcal{C} \left(h_0 x \right)$ for $0<h_0<1$, it follows that
\begin{align*}
& \mathcal{C} \left(\widetilde{P}_S(i+1) \right) - \left( \mathcal{C} \left(h_0 \widetilde{P}_S(i+1) \right) + \mathcal{C} \left( P_R^*(i+2) \right) \right) \\
\geq & \mathcal{C} \left(P_S^*(i+1) \right) - \mathcal{C} \left( h_0 P_S^*(i+1) \right) -  \mathcal{C} \left(  P_R^*(i+2) \right) \geq  0;
\end{align*}
and (\ref{no_buffer_incre_case2_equa5}) is due to the concavity of $\log(x)$ and $\delta \leq \frac{P_S^*(i) - P_S^*(i+1)}{2}$. Therefore, this case cannot happen.

3) $P_S^*(i) \leq P_S^*(i+1)$ and $P_R^*(i+1)>P_R^*(i+2)$: Since $\log \left( \frac{1+x}{1+h_0 x} \right)$ is strictly increasing over $x>0$ for fixed $h_0$, $0 \leq h_0 <1$, it follows that
      \begin{align*}
      \log \left( \frac{1 + P_S^*(i+1)}{1 + h_0 P_S^*(i+1)} \right) &  \geq \log \left( \frac{1 + P_S^*(i)}{1 + h_0 P_S^*(i)} \right) \\
      & \geq \log \left( 1+ P_R^*(i+1) \right),
      \end{align*}
      where the last inequality is due to Proposition \ref{no_buffer_cons}. Thus, it is obtained that
\begin{align}
\mathcal{C} \left(P_S^*(i+1) \right) & \geq \mathcal{C} \left(h_0 P_S^*(i+1) \right) + \mathcal{C} \left( P_R^*(i+1) \right) \nonumber \\
 & > \mathcal{C} \left(h_0 P_S^*(i+1) \right) + \mathcal{C} \left( \widetilde{P}_R(i+2) \right), \label{no_buffer_increasing_case5}
\end{align}
where $\widetilde{P}_R(i+1) = \widetilde{P}_R(i+2) = \frac{P_R^*(i+1) + P_R^*(i+2)}{2}$. By keeping $P_S^*(i)$ and $P_S^*(i+1)$ unchanged and updating the relay power values as $\widetilde{P}_R(i+1)$ and $ \widetilde{P}_R(i+2)$, it is observed that the relay energy constraints are still satisfied, and the sum rate is improved, i.e.,
      \begin{align}
      &~~~~\widetilde{R}(i) + \widetilde{R}(i+1) \\
      & = \mathcal{C} \left(h_0 P_S^*(i) \right) + \mathcal{C} \left( \widetilde{P}_R(i+1) \right) \nonumber \\
       &~~~~~~~~~+ \mathcal{C} \left(h_0 P_S^*(i+1) \right)  + \mathcal{C} \left( \widetilde{P}_R(i+2) \right) \label{no_buffer_increasing_case4} \\
& > \mathcal{C} \left(h_0 P_S^*(i) \right) + \mathcal{C} \left( P_R^*(i+1) \right) \nonumber \\
       &~~~~~~~~~+ \mathcal{C} \left(h_0 P_S^*(i+1) \right) + \mathcal{C} \left( P_R^*(i+2) \right) \label{no_buffer_increasing_case3} \\
& = R^*(i) +R^*(i+1), \nonumber
      \end{align}
      where (\ref{no_buffer_increasing_case4}) is due to (\ref{no_buffer_increasing_case5}) and the fact that $\mathcal{C} \left(P_S^*(i) \right) \geq \mathcal{C} \left(h_0 P_S^*(i) \right) + \mathcal{C} \left( P_R^*(i+1) \right) > \mathcal{C} \left(h_0 P_S^*(i) \right) + \mathcal{C} \left( \widetilde{P}_R(i+1) \right)$, and (\ref{no_buffer_increasing_case3}) is due to the concavity of $\log(x)$. Thus, this case cannot happen.

To summarize, since all the above three cases cannot be true, the only remaining case of $P_S^*(i) \leq P_S^*(i+1)$ and $P_R^*(i+1) \leq P_R^*(i+2)$ must be true. Proposition \ref{no_buffer_increasing_profile} is thus proved.

\section{The Optimality Proof of Algorithm \ref{no_buffer_algorithm}} \label{proof_no_optimal}

First, we prove that the source power profile $P_S^*(i)$'s obtained using Algorithm \ref{no_buffer_algorithm} are optimal for Problem (P1). Given $P_S^*(i)$'s, assume that the equalities of the source energy constraints are achieved at blocks $i_1,~i_2,\cdots,i_m=N$, while those are not achieved for the other blocks. Moreover, we assume that before block $i_s,~0 \leq s \leq m$ (define $i_0=0$), the optimal solution $P_S^{\star}(i)$'s of Problem (P1) are the same as $P_S^*(i)$'s, and their difference first appears at the $i$-th block, $i_s< i \leq i_{s+1}$. Due to this difference, the index of the next source energy exhausting block is denoted as $\widetilde{i}_{s+1}$, which may not be equal to $i_{s+1}$. Then, only three scenarios shown in Proposition \ref{no_buffer_threecase} may happen for both $P_S^*(i)$'s and $P_S^{\star}(i)$'s for the $(i_s+1)$-th to the $i_{s+1}$-th blocks and the $\widetilde{i}_{s+1}$-th block, respectively, which are discussed as follows:

\begin{figure*}[!b]
\vspace{4pt}
\hrulefill
\normalsize
\setcounter{Mytempeqncnt1}{\value{equation}}
\setcounter{equation}{70}
\begin{align}
& \left( l_1   \mathcal{C} \left( h_0 P_0\right)   + l_2 \mathcal{C} \left( P_0 + \frac{1}{h_0} - 1 \right) \right) -\left( l_1   \mathcal{C} \left( h_0 \left( P_0 + \delta \right)\right)  + l_2 \mathcal{C} \left( P_0 + \frac{1}{h_0} - 1 -\frac{l_1}{l_2} \delta \right) \right) >0, \label{no_buffer_proof_case1}
\end{align}
\setcounter{equation}{\value{Mytempeqncnt1}}
\end{figure*}

\begin{enumerate}[(I)]
  \item Scenario I happens for $P_S^*(i)$'s, and consider the following two cases:

\begin{enumerate}[(a)]
  \item If $P_S^{\star} (i)>P_S^*(i)$: By proposition \ref{no_buffer_increasing_profile}, it follows that $P_S^{\star} (j) \geq P_S^{\star} (i) > P_S^*(i) = P_S^*(j),~j=i,\cdots,i_{s+1}$, and then $B \sum_{k=1}^{i_{s+1}} P_S^{\star} (k) > B \sum_{k=1}^{i_{s+1}} P_S^*(k) = \sum_{k=1}^{i_{s+1}} E_S(k) $, which violates the source energy constraint. Thus, this case cannot occur.
  \item If $P_S^{\star} (i) < P_S^*(i)$:  For the case that Scenario I or II happens for $P_S^{\star} (i)$'s, it is easy to check that there will be no block $\widetilde{i}_{s+1} \geq i$ where the source energy is exhausted, and thus $P_S^{\star} (i)$'s violate Proposition \ref{no_buffer_source_energy}. For the case that Scenario III happens for $P_S^{\star} (i)$'s, it is observed that both the $\widetilde{k}_0$-th source and relay energy constraints in (\ref{source_relay_energy-const1}) and (\ref{source_relay_energy-const2}) are not achieved with equality, where $\widetilde{k}_0$ is the index of the source power transition block shown in Scenario III of Proposition \ref{no_buffer_threecase}. Then, there must exist $0<\theta<1$, such that the new source and relay power allocation profiles defined below still satisfy the source and relay energy constraints: $\widetilde{P}_S(\widetilde{k}_0) = \theta P_S^{\star}(\widetilde{k}_0) + ( 1-\theta) P_S^{\star}(\widetilde{k}_0+1) $, $\widetilde{P}_S(\widetilde{k}_0 + 1) = (1-\theta) P_S^{\star}(\widetilde{k}_0) + \theta P_S^{\star}(\widetilde{k}_0+1) $, $\widetilde{P}_R(\widetilde{k}_0+1) = \theta P_R^{\star}(\widetilde{k}_0+1) + ( 1-\theta) P_R^{\star}(\widetilde{k}_0+2) $, and $\widetilde{P}_R(\widetilde{k}_0 + 2) = (1-\theta) P_R^{\star}(\widetilde{k}_0+1) + \theta P_R^{\star}(\widetilde{k}_0+2) $. It is easy to check that with the new power allocation profiles, the sum rate over these two blocks is improved, since the rate function (\ref{orth_DF_rate}) is concave. Thus, this case cannot occur.
\end{enumerate}

\item Scenario III happens for $P_S^*(i)$'s, and assume that $P_S^*(j)=P_0,~j=i_s+1,\cdots,k_0$, and $P_S^*(j)=P_0 + \frac{1}{h_0} -1,~j=k_0+1,\cdots,i_{s+1}$. Consider the following two cases:

\begin{enumerate}[(a)]

\item If $P_S^{\star} (i)>P_S^*(i)$: When $i > i_{s,0}$, where $i_{s,0}$ is given by (\ref{no_buffer_posi}), it is easy to check that only Scenario III can happen for $P_S^{\star} (i)$'s from the $(i_s+1)$-th to the $\widetilde{i}_{s+1}$-th blocks. However, the source energy constraint at the $i_{s+1}$-th block is violated, since $P_S^{\star} (i)$'s are non-decreasing by Proposition \ref{no_buffer_increasing_profile}. As such, we only consider the case of $i \leq i_{s,0}$, which consists of two subcases:

(1) $P_S^*(i) < P_S^{\star} (i) < \widetilde{P}_S^{i,0}$, where $\widetilde{P}_S^{i,0}$ is given by (\ref{no_buffer_pswide1}) and (\ref{no_buffer_pswide2}): For $P_S^{\star} (i)$'s, similar to case (Ib), it follows that only Scenario III can happen, and denote $\widetilde{k}_0$ as the index of the source power transition block defined in Scenario III of Proposition \ref{no_buffer_threecase}. If $\widetilde{k}_0 \leq k_0$, by a similar argument as case (Ia), the source energy constraint at the $i_{s+1}$-th block will be violated; if $\widetilde{k}_0 > k_0$, by a similar argument of case (Ib), it follows that $P_S^{\star} (i)$'s cannot be optimal.

(2) $P_S^{\star} (i) = \widetilde{P}_S^{i,0}$: First, it is claimed that from the $(i_{s,0}+1)$-th to the $\widetilde{i}_{s+1}$-th blocks, there is no such index $\widetilde{k}_0$ corresponding to the source power transition block for Scenario III in Proposition \ref{no_buffer_threecase}. This is proved by contradiction, and consider the cases of $\widetilde{k}_0 > k_0$ and $\widetilde{k}_0 \leq k_0$ following the same argument as case (IIa1), respectively.

As such, for $P_S^{\star} (i)$'s, it is obtained that Scenario II happens from the $(i_s+1)$-th to the $k_0$-th blocks, and Scenario I happens from the $(k_0+1)$-th to the $i_{s+1}$-th blocks. Next, we prove that $P_S^{\star} (i)$'s are strictly sub-optimal over these blocks. Define a new source power profile as $\overline{P}_S(j) = \frac{\sum_{k=i_s+1}^{k_0} E_S(k)} {(k_0-i_s)B},~j=i_s+1,\cdots,k_0$ and $\overline{P}_S(j) = \frac{\sum_{k=k_0+1}^{i_{s+1}} E_S(k)} {(i_{s+1}-k_0)B}, ~j=k_0+1,\cdots,i_{s+1}$. Since both $\log (x)$ and $\log(h_0 x)$ are concave and some of the source energy constraints at the $(k_0+1)$-th to the $i_{s+1}$-th blocks may be violated by $\overline{P}_S(j)$'s, the sum rate $\sum_{j=i_s+1}^{i_{s+1}} R(i)$ given by $P_S^{\star} (i)$'s is upper-bounded by
\begin{align}
\sum_{j=i_s+1}^{i_{s+1}} R(i) & \leq \sum_{j=i_s+1}^{k_0} \mathcal{C} \left( h_0 \overline{P}_S(j)\right) + \mathcal{C} \left( \overline{P}_R (j+1)\right)  \nonumber \\
&~~~~+  \sum_{j=k_0 +1}^{i_{s+1}}  \mathcal{C} \left( \overline{P}_S(j) \right), \label{proof_new}
\end{align}
where $\overline{P}_R (j+1) = \widetilde{P}_R^{i+1,p}$, with $i_{r,p-1}<j \leq i_{r,p}$, $p \geq 0$, $i_{r,-1} = i_s$, $j \leq k_0$, and $i_{r,p}$'s and $\widetilde{P}_R^{i+1,p}$'s are given by (\ref{no_buffer_posi}), (\ref{no_buffer_pswide1}), (\ref{no_buffer_pswide2}), (\ref{search_sceIII_power2}), (\ref{search_sceIII_power31}), and (\ref{search_sceIII_power32}).

On the other hand, the sum rate for the $(i_s+1)$-th to $i_{s+1}$-th source messages with $P_S^*(i)$'s is
\begin{align}
\sum_{j=i_s+1}^{i_{s+1}} R_i^* & = \sum_{j=i_s+1}^{k_0}  \mathcal{C} \left( h_0 P_S^*(j)\right) + \mathcal{C} \left( P_R^*(j+1)\right) \nonumber \\
 & ~~~~~~~~~~+  \sum_{j=k_0 +1}^{i_{s+1}}  \mathcal{C} \left( P_S^*(j) \right), \label{proof_old}
\end{align}
where $P_R^*(j+1) = \widetilde{P}_R^{i+1,p}=\overline{P}_R (j+1)$ is due to (\ref{search_sceIII_power_final2}). To prove that (\ref{proof_old}) is larger than the right hand side of (\ref{proof_new}), it is equivalent to show that (\ref{no_buffer_proof_case1}) is true, where $\delta = \overline{P}_S(i_s+1)- P_S^*(i_s+1) = \overline{P}_S(i_s+1)- P_0 $, $l_1= k_0 - i_s$, and $l_2 = i_{s+1}- k_0$. By removing the $\log$ operations, (\ref{no_buffer_proof_case1}) can be rewritten as
        \setcounter{equation}{71}
        \begin{align} \label{no_proof_inequa}
        \frac{\left(1 + h_0 P_0 \right)^{l_1} \left(  P_0 + \frac{1}{h_0} \right)^{l_2} }{\left(1 + h_0 \left( P_0 + \delta \right) \right)^{l_1} \left( P_0 + \frac{1}{h_0} - \frac{l_1}{l_2}\delta \right)^{l_2} }>1.
        \end{align}
Let $g(x) = \left(1 + h_0 \left( P_0 + x \right) \right)^{l_1} \left( P_0 +  \frac{1}{h_0} - \frac{l_1}{l_2}x \right)^{l_2} $. Note that for $0 < x \leq \delta$, $ P_0 +  \frac{1}{h_0}  - \frac{l_1}{l_2}x >0$, and thus it follows that
        \begin{align*}
        & g'(x) = - l_1 \left(1 + h_0 \left( P_0 + x \right) \right)^{l_1-1} \\ & ~~~\cdot \left(  P_0 + \frac{1}{h_0} - \frac{l_1}{l_2}x \right)^{l_2-1} \left( \frac{l_1}{l_2} +1 \right) h_0 x <0,
        \end{align*}
which means that $g(x)$ is decreasing over $0 < x \leq \delta$, and it follows that (\ref{no_proof_inequa}) is true. Then, $P_S^{\star} (i)$'s cannot be optimal for this case.

\item If $P_S^{\star} (i) < P_S^*(i)$: First, note that if $i>i_s+1$, $P_S^{\star} (i)$'s violate Proposition \ref{no_buffer_increasing_profile}. Thus, we assume that $i=i_s+1$. Then, for $P_S^{\star}(i)$'s, if Scenario I or II happens, it is easy to check that the source energy cannot be completely consumed at the end of each $N$-block transmission, which violates Proposition \ref{no_buffer_source_energy}. Consider now the case that Scenario III happens. Denote the index of the source power transition block defined in Scenario III of Proposition \ref{no_buffer_threecase} as $\widetilde{k}_0$ and there are then two subcases: (i) If $\widetilde{k}_0 \geq k_0$, it is easy to check that this case violates Proposition \ref{no_buffer_source_energy}, since there will be no source energy exhausting blocks after the $k_0$-th block; (ii) if $\widetilde{k}_0 < k_0$, it is easy to check that the relay power constraint at the $\widetilde{k}_0$-th block is achieved with equality with a similar argument as case (Ib) under Scenario III. Moreover, it can be shown that $P_S^{\star}(j) > P_S^*(j)$, $ \widetilde{k}_0 <j \leq k_0$ (if not, there will be no source power energy exhausting block existed after the $k_0$-th block). Then, it follows that $P_R^{\star}(j+1) \geq \frac{(1-h_0) P_S^{\star}(j) } {1+ h_0 P_S^{\star}(j)} >  \frac{(1-h_0) P_S^*(j) } {1+ h_0 P_S^*(j)} \geq P_R^*(j+1) $ for $ \widetilde{k}_0 <j \leq k_0$. As such, it is observed that the the relay power constraint at the $k_0$-th block is violated. Therefore, this case cannot occur.
\end{enumerate}

\item Scenario II happens for $P_S^*(i)$'s, and consider the following two cases:

\begin{enumerate}[(a)]

\item If $P_S^{\star} (i)>P_S^*(i)$: Similar to case (Ia), it follows that $P_S^{\star}(i)$'s cannot be optimal.

\item If $P_S^{\star}(i) < P_S^*(i)$: For $P_S^{\star}(i)$'s, Scenario I and II cannot occur due to the same argument as case (Ib); for the case that Scenario III happens for $P_S^{\star}(i)$'s, it cannot happen, since by using Algorithm \ref{no_buffer_algorithm}, we cannot find such $k_0$ and $\widehat{P}_S^i$ satisfying (\ref{search_sceIII_k01}), (\ref{search_sceIII_k02}), (\ref{search_sceIII_power11}), and (\ref{search_sceIII_power22}), respectively. Then, we only need to show that (\ref{search_sceIII_k01}), (\ref{search_sceIII_k02}), (\ref{search_sceIII_power11}), and (\ref{search_sceIII_power22}) are necessary for the existence of Scenario III. By a similar argument as case (II), it can be shown that (\ref{search_sceIII_k01}) and (\ref{search_sceIII_k02}) are necessary for the existence of $k_0$. For (\ref{search_sceIII_power11}) and (\ref{search_sceIII_power22}), it is shown as follows: (i) $\widehat{P}_S^i \leq \widetilde{P}_S^{i,0}$: this is due to the source energy constraint at the $i_{s,0}$-th block; (ii) $\widehat{P}_S^i \geq P_S^*(i_s) $: this is due to Proposition \ref{no_buffer_increasing_profile}; (iii) $\widehat{P}_S^i \geq \frac{ \widetilde{P}_R(k_0+1)} { 1- h_0 - h_0 \widetilde{P}_R(k_0+1) }$, which is equivalent to that $ \widetilde{P}_R(k_0+1) \leq  \frac{(1-h_0)\widehat{P}_S^i} {1+ h_0 \widehat{P}_S^i }$: If this condition is not true, which means that the $k_0$-th source and relay energy constraints in (\ref{source_relay_energy-const1}) and (\ref{source_relay_energy-const2}) are not achieved with equality, it can be shown that this case is not optimal with the same argument of case (Ib) under Scenario III; (iv) $ \widehat{P}_S^i +  \frac{1}{h_0} - 1 \leq \frac{ \widetilde{P}_R(k_0+2) }{ 1- h_0 - h_0 \widetilde{P}_R(k_0+2)} $, which is equivalent to that $ \widetilde{P}_R(k_0+2) \geq \frac{(1-h_0) \left( \widehat{P}_S^i +  \frac{1}{h_0} - 1  \right)} {1+ h_0 \left( \widehat{P}_S^i +  \frac{1}{h_0} - 1  \right)}$: If this condition is not true, it will violate the optimality conditions for Scenario III shown in Proposition \ref{no_buffer_threecase}. Therefore, the necessity of (\ref{search_sceIII_k01}), (\ref{search_sceIII_k02}), (\ref{search_sceIII_power11}), and (\ref{search_sceIII_power22}) is proved.
\end{enumerate}

\end{enumerate}

In conclusion, $P_S^*(i)$'s are optimal for Problem (P1). Now, with the optimal $P_S^*(i)$'s, we know which scenario in Proposition \ref{no_buffer_threecase} happens for each block. It is thus easy to show that the corresponding relay power profile $P_R^*(i+1)$'s obtained using Algorithm \ref{no_buffer_algorithm} are also optimal: If $P_R^*(i+1)$'s are given by (\ref{no_buffer_water1}), they are optimal obviously; if $P_R^*(i+1)$'s are given by (\ref{no_buffer_water2}), they are also optimal by proving that the two cases $P_R^{\star}(i+1) > P_R^*(i+1)$ and $P_R^{\star}(i+1) < P_R^*(i+1)$ both cannot occur, for which the proof is similar to that in case (I) and thus omitted. Therefore, the optimality of Algorithm \ref{no_buffer_algorithm} is proved.

\section{Proof of Proposition \ref{with_source_profile_opt}}
\label{with_source_opt}

Define $f(x)= \frac{1}{2} \log \left( \frac{1+x}{1+h_0x} \right)$ over $x>0$ with fixed $0 \leq h_0 <1 $. Since $f''(x) = - \frac{1}{2} \left( 1+x \right)^{-2} + \frac{1}{2} h_0^2 \left( 1+h_0 x  \right)^{-2}$ and $\frac{h_0}{1+h_0x} < \frac{1}{1 + x}$, it follows that $f''(x)<0$ and $f(x)$ is concave. Moreover, it is easy to check that $f(x)$ is increasing over $x>0$. Then, we obtain the following lemma.
\begin{Lemma} \label{appen_lemma}
For the case of $0 < h_0 <1$, the optimal source power profile $P_S^*(i)$'s of Problem (P2) are non-decreasing over $i$; for $h_0 =0$, there exist optimal source power profile $P_S^*(i)$'s, which are non-decreasing over $i$, $i=1,\cdots,N$.
\end{Lemma}
\begin{proof}
Denote the optimal solution of Problem (P2) as $P_S^{\star} (i)$'s and $P_R^{\star} (i+1)$'s. For the case of $0 < h_0 <1$, consider any consecutive source/relay power pair corresponding to the $i$-th and the $(i+1)$-th source messages, $i=1,\cdots,N-1$, and the following two cases:

1) $P_S^{\star} (i) > P_S^{\star} (i+1)$, $P_R^{\star} (i+1) > P_R^{\star} (i+2)$: Define $r^{\star} (i+1) = \mathcal{C}(P_R^{\star} (i+1))$. Since $\log(x)$ is increasing, it follows that $r^{\star} (i+1) > r^{\star} (i+2)$. Then, the constraint in (\ref{inf_buffer_DF2}) is equivalent to $\sum_{k=1}^i r^{\star} (k+1) \leq \sum_{k=1}^i f(P_S^{\star}  (k)) $, which is convex. Thus, similar to case 1) of Appendix \ref{proof_no_increase_profile}, this case is not optimal for Problem (P2).

2) $P_S^{\star} (i) > P_S^{\star} (i+1)$, $P_R^{\star}  (i+1) \leq P_R^{\star}  (i+2)$: It is first proved that the constraint $\sum_{k=1}^i \mathcal{C}(P_R^{\star}  (k+1)) \leq \sum_{k=1}^i f(P_S^{\star}  (k)) $ is not satisfied with equality in this case by contradiction as follows. Suppose that the above rate inequality is satisfied with equality. From the $(i-1)$-th and the $(i+1)$-th constraints in (\ref{inf_buffer_DF2}), it follows that $ \mathcal{C}(P_R^{\star} (i+1)) \geq f(P_S^{\star} (i)) $ and $ \mathcal{C}(P_R^{\star} (i+2)) \leq f(P_S^{\star} (i+1)) $, and together with the assumption that $P_R^{\star}  (i+1) \leq P_R^{\star}(i+2)$, it follows that $  f(P_S^{\star} (i)) \leq \mathcal{C}(P_R^{\star} (i+1)) \leq \mathcal{C}(P_R^{\star} (i+2)) \leq f(P_S^{\star} (i+1)) $. Since $f(x)$ is an increasing function, it follows that $P_S^{\star}(i) \leq P_S^{\star}(i+1)$, which contradicts the assumption that $P_S^{\star} (i) > P_S^{\star} (i+1)$.

Then, there must exist $\delta>0$, with which we define a new power allocation as $\widetilde{P}_S(i) = P_S^{\star}(i) -\delta$ and $ \widetilde{P}_S(i+1) = P_S^{\star}(i+1)+\delta$, while keeping $P_R^{\star}(i+1)$ and $P_R^{\star}(i+2)$ unchanged. It is easy to check that the new power allocation still satisfies the constraints in Problem (P3) and increases the sum rate over these two blocks. Thus, case 2) is not optimal for Problem (P2).

In conclusion, for the case of $0 < h_0 <1$, only the case $P_S^{\star} (i) \leq P_S^{\star} (i+1)$ and $P_R^{\star} (i+1) \leq P_R^{\star}  (i+2)$ can be optimal. For the case of $h_0 =0$, similar argument can be applied to show the existence of such a non-decreasing optimal solution, which is omitted for brevity. Thus, this lemma is proved.
\end{proof}

Next, we prove Proposition \ref{with_source_profile_opt}. First, consider the case of $0 < h_0 <1$. For the optimal source power $P_S^{\star}(i)$'s of Problem (P2) and $P_S^*(i)$'s obtained using Algorithm \ref{with_source_profile}, we assume that $P_S^{\star} (j) = P_S^*(j)$, $j=1,\cdots,i-1$, and $P_S^{\star}(i) \neq P_S^*(i)$. Consider the following two cases:

1) $P_S^{\star} (i)>P_S^*(i)$: Similar to the proof of case (Ia) in Appendix \ref{proof_no_optimal}, it follows that this case violates the source energy constraints, since $P_S^{\star}(i)$'s are non-decreasing, and thus cannot be true.

2) $P_S^{\star} (i)<P_S^*(i)$: First, we claim the following two results, which will be proved later: (i) $P_S^{\star} (i)$'s reduce the optimal value of Problem (P2) without considering the relay power allocation; (ii) By further considering the relay power allocation and the constraint (\ref{inf_buffer_DF2}), $P_S^{\star}(i)$'s shrink the feasible set of $P_R(i+1)$'s. If (i) and (ii) are both true, it follows that $P_S^{\star} (i)<P_S^*(i)$ cannot be true.

For the proof of (i), from \cite{yang}, it follows that $P_S^*(i)$'s maximize $\sum_{k=1}^N \mathcal{C} \left(h_0 P_S (k) \right)$ subject to the constraint (\ref{no_constr_source_cons}); in other words, for the case of $P_S^{\star} (i)<P_S^*(i)$, this value will be strictly decreased. Thus, (i) is proved.

For the proof of (ii), it is equivalent to prove that $\sum_{k=1}^j f(P_S^{\star}(k)) \leq \sum_{k=1}^j f(P_S^*(k))$, $i \leq j \leq N $. When $j=i$, by the assumption that $P_S^{\star} (i)< P_S^*(i)$, it follows that $\sum_{k=1}^j f(P_S^{\star} (k)) \leq \sum_{k=1}^j f(P_S^*(k))$. For $j=i+1$, if $\sum_{k=1}^j f(P_S^{\star}  (k)) > \sum_{k=1}^j f(P_S^*(k))$, it follows that $f(P_S^{\star} (j)) > f(P_S^*(j))$ and thus $P_S^{\star} (j) > P_S^*(j)$ since $f(x)$ is monotonically increasing. For $P_S^*(i)$'s, denote the index of the next source energy exhausted block after the $j$-th block as $j_s$. It is then obtained that $P_S^{\star} (k) \geq P_S^{\star} (j) > P_S^*(j)= P_S^*(k)$, $j \leq k \leq j_s$, since by Lemma \ref{appen_lemma}, $P_S^{\star} (i)$'s are non-decreasing over $i$. Moreover, by \cite{yang}, since $f(x)$ is concave, it can be shown that with the same energy budget over the first to the $j$-th blocks, if $\sum_{k=1}^j f(P_S^{\star} (k)) > \sum_{k=1}^j f(P_S^*(k))$, it follows that $\sum_{k=1}^j P_S^{\star} (k) > \sum_{k=1}^j P_S^*(k)$. Then, it is easy to check that $P_S^{\star}(j)$'s violate the source energy constraint at the $j_s$-th block, i.e., $B \sum_{k=1}^{j_s} P_S^{\star} (k) > B \sum_{k=1}^{j_s} P_S^*(k ) = \sum_{k=1}^{j_s} E_S(k)$. Thus, it is obtained that $\sum_{k=1}^j f(P_S^{\star} (k)) \leq \sum_{k=1}^j f(P_S^*(k))$ for $j=i+1$. By using the mathematical induction method, it can be shown that $\sum_{k=1}^j f(P_S^{\star} (k)) \leq \sum_{k=1}^j f(P_S^*(k))$, $1 \leq j \leq N $, which suggests that (ii) is true.

From the above analysis, it is proved that for the case of $0 < h_0 <1$, $P_S^*(i)=P_S^{\star}(i)$, $i=1,\cdots,N$, and thus $P_S^*(i)$'s are optimal for Problem (P2). For the case of $h_0 = 0$, since there always exists one optimal source power profile that is non-decreasing over $i$ (by Lemma \ref{appen_lemma}), without loss of generality, we can assume that $P_S^{\star} (i)$'s are still non-decreasing over $i$. With the same argument as the case of $0 <h_0 <1$, it is observed that claim (ii) is still true, which means that the maximum value with $P_S^{\star}(i)$'s is not larger than that with $P_S^*(i)$'s. In conclusion, this proposition is proved.

\section{Proof of Proposition \ref{compare_two}} \label{proof_compare}

Denote $P_S^{\star} (i)$'s and $P_R^{\star} (i+1)$'s as the optimal solution for Problem (P1), $P_S^*(i)$'s and $P_R^*(i+1)$'s as the optimal solution for Problem (P2). First, we prove that the condition given in Proposition \ref{compare_two} is sufficient. Assume that $\mathcal{C}(P_R^*(k+1)) \leq \mathcal{C}(P_S^*(k)) - \mathcal{C}( h_0 P_S^*(k))$, $1 \leq k \leq i$, and $\mathcal{C}(P_R^*(i+2)) > \mathcal{C}(P_S^*(i+1)) - \mathcal{C}( h_0 P_S^*(i+1))$. Due to constraint (\ref{inf_buffer_DF2}), without loss of generality, we further assume that $\mathcal{C}(P_R^*(i+1)) < \mathcal{C}(P_S^*(i)) - \mathcal{C}( h_0 P_S^*(i))$. Then, for $P_S^{\star} (i)$'s and $P_R^{\star} (i+1)$'s, it is easy to check that only Scenario II or III given in Proposition \ref{no_buffer_threecase} can happen at the $i$-th block, which is further discussed as follows:

1) Scenario III happens for the $i$-th source message in Problem (P1), i.e., $P_S^{\star}(i)= P_0$, $P_S^{\star}(i+1) = P_0+\frac{1}{h_0} - 1$, and thus the source energy constraint at the $i$-th block is not satisfied. Then, there exists $\delta>0$, such that the newly defined source power allocation $\widetilde{P}_S(i)= P_0 +\delta$, $\widetilde{P}_S(i+1) = P_0+\frac{1}{h_0} - 1 -\delta$ satisfies the source energy constraint at the $i$-th block. Moreover, as for the NDC case, we can increase the binning rate of $i$-th source message with the amount $\mathcal{C}(P_S^{\star}(i+1)) - \mathcal{C}(\widetilde{P}_S(i+1))$ (note that this operation is possible since this amount is less than $ \mathcal{C}(\widetilde{P}_S(i)) - \mathcal{C}(P_S^{\star}(i)) $), and transmit it at the $(i+1)$-th relay transmission. With the above scheme, it is easy to check that the sum rate over these two blocks is strictly improved.

2) Scenario II happens at the $i$-th block for Problem (P1). Consider two subcases: (a) If the relay energy is not exhausted at the end of each $N$-block transmission, as for the NDC case, we can increase $R_B(i+1)$, and use the available relay energy at the $(N+1)$-th block to transmit the increased part in $R_B(i+1)$, which strictly improves the throughput of the DC case; (b) If the relay energy is exhausted, there must exists $k>i+2$, such that $P_R^{\star} (k) > P_R^{\star} (i+2)$. Thus, there exists $0 < \delta < \frac{P_R^{\star} (k) - P_R^{\star} (i+2)}{2}$. Define a new relay power allocation satisfying the relay energy constraint as $ \widetilde{P}_R(i+2) =  P_R^{\star} (i+2) + \delta$ and $ \widetilde{P}_R(k) =  P_R^{\star} (k) - \delta$. By increasing the binning rate of the $i$-th message with the amount $\mathcal{C}(\widetilde{P}_R(i+2)) - \mathcal{C}( P_R^{\star} (i+2)) $ and decreasing that of the $k$-th message with $ \mathcal{C}( P_R^{\star} (k)) - \mathcal{C}(\widetilde{P}_R(k)) $, it can be shown that the new scheme improves the sum rate of these two blocks, since $\mathcal{C}(\widetilde{P}_R(i+2)) - \mathcal{C}( P_R^{\star} (i+2)) > \mathcal{C}( P_R^{\star} (k)) - \mathcal{C}(\widetilde{P}_R(k))$. Based on 1) and 2), the ``if'' part is proved.

Next, we prove that the condition in Proposition \ref{compare_two} is also necessary. Assume that $\mathcal{C}(P_R^*(i+1)) \leq \mathcal{C}(P_S^*(i)) - \mathcal{C}( h_0 P_S^*(i))$, $\forall i \in \{1,\cdots, N\}$. It is easy to check that with the power allocation $P_S^*(i)$'s and $P_R^*(i+1)$'s of Problem (P2), the achievable rate of the $i$-th block for Problem (P1) is given as $R(i) = \min \left\{\mathcal{C} \left( P_S^*(i) \right), \mathcal{C} \left( P_R^*(i+1) \right) +\mathcal{C} \left( h_0 P_S^*(i) \right) \right\} = \mathcal{C} \left( P_R^*(i+1) \right) +\mathcal{C} \left( h_0 P_S^*(i) \right)$, which equals that of the NDC case. By searching over the whole feasible set, the throughput of the DC case will be no smaller than that of the NDC case. Together with Proposition \ref{compare_two_schemes}, it is obtained that the throughput of the two cases are identical, and thus the ``only if'' part is proved.

\end{document}